\pdfoutput=1
\documentclass{toc}
\tocdetails{%
volume=14, number=9, year=2018, firstpage=1,
specissue={\cccBG},  %
received={June 16, 2016},
revised={May 10, 2018},
published={June 5, 2018},
doi={10.4086/toc.2018.v014a009},
title={Average-Case Lower Bounds and Satisfiability Algorithms for Small Threshold Circuits},
author={Ruiwen Chen, Rahul Santhanam, and Srikanth Srinivasan},
plaintextauthor={Ruiwen Chen, Rahul Santhanam, Srikanth Srinivasan},
acmclassification={F.1.2, F.2.3},
amsclassification={68Q17, 68Q32},
keywords={complexity theory, circuit complexity, correlation bounds, threshold functions, random restrictions, learning, SAT}
}

\begin{document}

\begin{frontmatter}

\title{Average-Case Lower Bounds and\\ Satisfiability Algorithms for\\ Small Threshold Circuits\titlefootnote{A conference version of this paper appeared in the Proceedings of the 31st Conference on Computational Complexity, 2016~\cite{CSS-conf}.}}

  \author[chen]{Ruiwen Chen}
\author[santhanam]{Rahul Santhanam}
\author[srinivasan]{Srikanth Srinivasan}

\begin{abstract}
We show average-case lower bounds for explicit Boolean functions against bounded-depth threshold circuits with a superlinear number of wires. We show that for each integer $d > 1$, there is a constant $\varepsilon_d > 0$ such that the Parity function on $n$ bits has correlation at most $n^{-\varepsilon_d}$ with depth-$d$ threshold circuits which have at most $n^{1+\varepsilon_d}$ wires, and the Generalized Andreev function on $n$ bits has correlation at most $\exp(-{n^{\varepsilon_d}})$ with depth-$d$ threshold circuits which have at most $n^{1+\varepsilon_d}$ wires. Previously, only worst-case lower bounds in this setting were known (Impagliazzo, Paturi, and Saks (SICOMP 1997)).

We use our ideas to make progress on several related questions. We give satisfiability algorithms beating brute force search for depth-$d$ threshold circuits with a superlinear number of wires. These are the first such algorithms for depth greater than 2. We also show that Parity on $n$ bits cannot be computed by polynomial-size $\textsf{AC}^0$ circuits with $n^{o(1)}$ \emph{general} threshold gates. Previously no lower bound for Parity in this setting could handle more than $\log(n)$ gates. This result also implies subexponential-time learning algorithms for $\textsf{AC}^0$ with $n^{o(1)}$ threshold gates under the uniform distribution. In addition, we give almost optimal bounds for the number of gates in a depth-$d$ threshold circuit computing Parity on average, and show average-case lower bounds for threshold formulas of \emph{any} depth. 

Our techniques include adaptive random restrictions, anti-concentration and the structural theory of linear threshold functions, and bounded-read Chernoff bounds. 
\end{abstract}

%

\end{frontmatter}

\section{Introduction}

\newcommand{\wtOmega}{\widetilde{\Omega}}

\newcommand{\MODm}{\mathrm{MOD}_m}    %
\newcommand{\MODp}{\mathrm{MOD}_p}    %

\newcommand{\AC}{\cclass{AC}}
\newcommand{\THR}{\mathrm{THR}}
\newcommand{\AND}{\mathrm{AND}}
\newcommand{\NS}{\mathrm{NS}}
\newcommand{\Var}{\mathrm{Var}}
\newcommand{\TAC}{\cclass{TAC}}

\newcommand{\MS}{\mathrm{MS}}
\newcommand{\ip}[2]{\langle #1, #2\rangle}
\newcommand{\mc}[1]{\mathcal{#1}}
\newcommand{\prob}[2]{\mathop{\mathrm{Pr}}_{#1}\left[#2\right]}
\newcommand{\avg}[2]{\mathop{\mathbf{E}}_{#1}\left[#2\right]}
\newcommand{\sgn}{\mathrm{sgn}}
\newcommand{\Corr}{\mathrm{Corr}}
\newcommand{\poly}{\mathrm{poly}}
\newcommand{\mynorm}[1]{\lVert #1\rVert}
\newcommand{\fanin}{\mathrm{fan}\text{-}\mathrm{in}}
\newcommand{\Par}{\mathrm{Par}}

\makeatletter
\newtheorem*{rep@theorem}{\rep@title}
\newcommand{\newreptheorem}[2]{%
\newenvironment{rep#1}[1]{%
 \def\rep@title{#2 \ref{##1}}%
 \begin{rep@theorem}}%
 {\end{rep@theorem}}}
\makeatother

\newreptheorem{theorem}{Theorem}
\newreptheorem{lemma}{Lemma}
\newreptheorem{cor}{Corollary}

One of the main goals in complexity theory is to prove circuit lower bounds for explicit functions in \cclass{P} or \cclass{NP}. We seem quite far from being able to prove that there is a problem in \cclass{NP} that requires superlinear Boolean circuits. We have some understanding, via formulations such as the relativization barrier~\cite{Baker-Gill-Solovay75}, the ``natural proofs'' barrier~\cite{Razborov-Rudich97} and the algebrization barrier~\cite{Aaronson-Wigderson08}, of why current techniques are inadequate for this purpose.

However, the community has had more success proving explicit lower bounds against bounded-depth circuits of various kinds. Thanks to pioneering work of Ajtai~\cite{Ajtai83}, Furst-Saxe-Sipser~\cite{Furst-Saxe-Sipser84}, Yao~\cite{Yao85} and H\r{a}stad~\cite{Hastad}, we know that the Parity and Majority functions require bounded-depth unbounded fan-in circuits of exponential size if only
AND, OR, and NEGATION   %
gates are allowed. Later,  Razborov~\cite{Razborov87} and Smolensky~\cite{Smolensky87} showed that Majority requires exponential size even when $\MODp$ gates are allowed in addition to AND and OR gates, for any prime $p$.\footnote{Recall that for an integer $m\geq 2$, a $\MODm$ gate accepts exactly those inputs whose Hamming weight is not divisible by $m$.} The case of bounded-depth circuits with AND, OR and $\MODm$ gates, where $m$ is a composite, has been open for nearly thirty years now, even though Majority is conjectured to be hard for such circuits. Williams~\cite{Williams11}  
recently made significant progress by showing that
nondeterministic     %
exponential time does not have
super-polynomial size circuits with AND, OR and $\MODm$ gates, for any $m$.

For all the bounded-depth circuit classes above, Majority is either known or conjectured to be hard. How about circuit classes which incorporate majority gates, or more generally, gates that are arbitrary linear threshold functions? Note that such gates generalize AND and OR, though not $\MODp$. In the 90s, there was some work on studying the power of bounded-depth threshold circuits.
Paturi and Saks~\cite{PS} showed that depth-2 circuits with majority gates computing Parity require ${\wtOmega}(n^2)$ wires; there is also a nearly matching upper bound for Parity. Impagliazzo, Paturi and Saks~\cite{IPS} considered bounded-depth threshold circuits with arbitrary linear threshold gates, and showed that for each depth $d$, there is a constant $\varepsilon_d > 0$ such that Parity requires $n^{1+\varepsilon_d}$ wires to compute with depth-$d$ threshold circuits.

These lower bounds are \emph{worst-case} lower bounds---they show that for any sequence of small circuits, there exist inputs of every length on which the circuits fail to compute Parity. There are several reasons to be interested in \emph{average-case} lower bounds under the uniform distribution, or equivalently, in \emph{correlation} upper bounds. For one, average-case lower bounds show that a randomly chosen input is likely to be hard, and thus give a way to generate hard instances efficiently. 
Second, average-case lower bounds are closely tied to pseudorandom generators via the work of
Nisan and Wigderson~\cite{nisan91cca, Nisan-Wigderson94},  %
and are indeed a prerequisite for obtaining pseudorandom generators with non-trivial seed length for a circuit class. Third, recent work on satisfiability algorithms~\cite{Santhanam10, Impagliazzo-Matthews-Paturi12, CKKSZ} indicates that the design and analysis of non-trivial satisfiability algorithms is closely tied to proving average-case lower bounds, though there is no formal connection. Fourth, the seminal
paper by  %
Linial, Mansour, and Nisan~\cite{Linial-Mansour-Nisan93}
shows that average-case lower bounds for Parity against a circuit class are tied to non-trivially learning the circuit class under the uniform distribution.

With these different motivations in mind, we systematically study average-case lower bounds for bounded-depth threshold circuits. Our first main result shows correlation upper bounds for Parity and another explicit function known as the Generalized Andreev function with respect to threshold circuits with few wires. No correlation upper bounds for explicit functions against bounded-depth threshold circuits with superlinear wires was known before our work.

\begin{theorem}
\label{MainThm1}
For each depth $d \geq 1$, there is a constant $\varepsilon_d > 0$ such that for all large enough $n$, no threshold circuit of depth $d$ with at most $n^{1+\varepsilon_d}$ wires agrees with Parity on more than $1/2+1/n^{\varepsilon_d}$ fraction of inputs of length $n$, and with the Generalized Andreev function on more than $1/2+1/2^{n^{\varepsilon_d}}$ fraction of inputs of length $n$.
\end{theorem}
 
\expref{Theorem}{MainThm1} captures the content of \expref{Theorem}{thm:corr-parity} and \expref{Theorem}{thm:genstrong} in Section 4. 

Quite often, ideas behind lower bound results for a circuit class $\mathcal{C}$ result in satisfiability algorithms (running in better than brute-force time) for $\mathcal{C}$, and vice-versa (see, \eg,~\cite{ppz}). Here, we leverage the ideas behind our correlation bounds against Threshold circuits to improved satisfiability algorithms for these circuits. More precisely, we constructivize the ideas of the proof of the strong correlation upper bounds for the Generalized Andreev function to get non-trivial satisfiability algorithms for bounded-depth threshold circuits with few wires. Previously, such algorithms were only known for depth-$2$ circuits, due to Impagliazzo--Paturi--Schneider~\cite{ImpagliazzoPaturiSchneider} and Tamaki (unpublished). 
 
 \begin{theorem}
 \label{MainThm2}
 For each depth $d \geq 1$, there is a constant $\epsilon_d > 0$ such that the satisfiability of depth-$d$ threshold circuits with at most $n^{1+\varepsilon_d}$ wires can be solved in randomized time $2^{n-n^{\varepsilon_d}} \poly(n)$.
 \end{theorem}
 
 \expref{Theorem}{MainThm2} is re-stated and proved as \expref{Theorem}{thm:sat-algorithm} in \expref{Section}{sec:satalgo}.
 
 Using our ideas, we also show correlation bounds against $\AC^0$ circuits with a few threshold gates, as well as learning algorithms under the uniform distribution for such circuits.
 
 \begin{theorem}
 \label{MainThm3}
For each constant $d$,  there is a constant $\gamma > 0$ such that Parity has correlation at most $1/n^{\Omega(1)}$ with $\AC^0$ circuits of depth $d$ and size at most $n^{\log(n)^{0.4}}$ augmented with at most $n^{\gamma}$ threshold gates. Moreover, the class of $\AC^0$ circuits of size at most $n^{\log(n)^{0.4}}$ augmented with at most $n^{\gamma}$ threshold gates can be learned to constant error under the uniform distribution in time $2^{n^{1/4+o(1)}}$. 
 \end{theorem}
 
 \expref{Theorem}{MainThm3} captures the content of \expref{Corollary}{cor:taco} and \expref{Theorem}{thm:learning-taco} in Section 7.
 
 Having summarized our main results, we now describe related work and our proof techniques in more detail.
 
\subsection{Related work}

There has been a large body of work proving upper and lower bounds  for constant-depth threshold circuits. Much of this work has focused on the setting of small gate complexity, which seems to be the somewhat easier case to handle. A distinction must also be drawn between work that has focused on the setting where the threshold gates are assumed to be \emph{majority gates} (\ie, the linear function sign representing the gate has integer coefficients that are bounded by a polynomial in the number of variables) and work that focuses on general threshold gates, since analytic tools such as \emph{rational approximation} that are available for majority gates do not work in the setting of general threshold gates.

We discuss the work on wire complexity first, followed by the results on gate complexity.

\paragraph{Wire complexity.} Paturi and Saks~\cite{PS} considered depth-$2$ \emph{majority} circuits and showed an ${\wtOmega}(n^2)$ lower bound on the wire complexity required to compute Parity; this nearly matches the upper bound of $O(n^2)$. They also showed that there exist majority circuits of size $n^{1+\Theta(\varepsilon_1^d)}$ and depth $d$ computing Parity; here $\varepsilon_1 = 2/(1+\sqrt{5})$. Impagliazzo, Paturi, and Saks~\cite{IPS} showed a depth-$d$ lower bound for \emph{general} threshold circuits computing Parity: namely, that any such circuit must have wire complexity at least $n^{1+\varepsilon_2^d}$ where $\varepsilon_2 < \varepsilon_1$.

The proof of~\cite{IPS} proceeds by induction on the depth $d$. The main technical lemma shows that a circuit of depth $d$ can be converted to a depth-$(d-1)$ circuit of the same size by setting some of the input variables. The fixed variables are set in a random fashion, but \emph{not} according to the uniform distribution. In fact, this distribution has statistical distance close to $1$ from the uniform distribution and furthermore, depends on the circuit whose depth is being reduced. Therefore, it is unclear how to use this technique to prove a correlation bound with respect to the uniform distribution. In contrast, we are able to reduce the depth of the circuit by setting variables uniformly at random (though the variables that we restrict are sometimes chosen in a way that depends on the circuit), which yields the correlation bounds we want.

\paragraph{Gate complexity.} The aforementioned
paper by  %
Paturi and Saks~\cite{PS} also proved a near optimal ${\wtOmega}(n)$ lower bound on the number of gates in any depth-$2$ majority circuits computing Parity. 

Siu, Roychowdhury, and Kailath~\cite{Siuetal} considered majority circuits of bounded depth and small gate complexity. They showed that Parity can be computed by depth-$d$ majority circuits with $O(dn^{1/(d-1)})$ gates. Building on the ideas of~\cite{PS}, they also proved a near matching lower bound of ${\wtOmega}(dn^{1/(d-1)})$. Further, they also considered the problem of correlation bounds and showed that there exist depth-$d$ majority circuits with $O(dn^{1/2(d-1)})$ gates that compute Parity almost everywhere and that majority circuits of significantly smaller size have $o(1)$ correlation with Parity (\ie, these circuits cannot compute Parity on more than a $1/2 + o(1)$ fraction of inputs; recall that $1/2$ is trivial since a constant function computes Parity correctly on $1/2$ of its inputs). Impagliazzo, Paturi, and Saks~\cite{IPS} extended the worst-case lower bound to general threshold gates, where they proved a slightly weaker lower bound of $\Omega(n^{1/2(d-1)})$. As discussed above, though, it is unclear how to use their technique to prove a correlation bound.

Beigel~\cite{Beigel} extended the result of Siu et al.\ to the setting of $\AC^0$ augmented with a few majority gates. He showed that any subexponential-sized depth-$d$ $\AC^0$ circuit with significantly less than some $k = n^{\Theta(1/d)}$ majority gates has correlation $o(1)$ with Parity. The techniques of all the above results with the exception of~\cite{IPS} were based on the fact that majority gates can be well-approximated by low-degree \emph{rational functions}. However, this is not true for general threshold functions~\cite{Sherstov} and hence, these techniques do not carry over to the case of general threshold gates.

A lower bound technique that does carry over to the setting of general threshold gates is that of showing that the circuit class has low-degree polynomial \emph{sign-representations}. Aspnes, Beigel, Furst and Rudich~\cite{ABFR} used this idea to prove that $\AC^0$ circuits augmented with a single general threshold \emph{output} gate---we refer to these circuits as $\TAC^0$ circuits as in~\cite{GS}---of subexponential-size and constant-depth have correlation $o(1)$ with Parity. More recently, Podolskii~\cite{Podolskii} used this technique along with a trick due to Beigel~\cite{Beigel} to prove similar bounds for subexponential-sized $\AC^0$ circuits augmented with general threshold gates. However, this trick incurs an exponential blow-up with the number of threshold gates and hence, in the setting of the Parity function, we cannot handle $k > \log n$ threshold gates.

Another technique that has proved useful in handling general threshold gates is \emph{Communication Complexity}, where the basic idea is to show that the circuit---perhaps after restricting some variables---has low communication complexity in some suitably defined communication model. We can then use results from communication complexity to infer lower bounds or correlation bounds. Nisan~\cite{Nisan-threshold} used this technique to prove exponential correlation bounds against general threshold circuits (not necessarily even constant-depth) with $n^{1-\Omega(1)}$ threshold gates. Using Beigel's trick and multiparty communication complexity bounds of Babai, Nisan and Szegedy~\cite{BNS}, Lovett and Srinivasan~\cite{LS} (see also~\cite{RW,HM}) proved exponential correlation bounds against any polynomial-sized $\AC^0$ circuits augmented with up to $n^{\frac{1}{2}-\Omega(1)}$ threshold gates. 

We do not use this technique in our setting for many reasons. Firstly, it cannot be used to prove lower bounds or correlation bounds for functions such as Parity (which has small communication complexity in most models). In particular, these ideas do not yield the noise sensitivity bounds we get here. Even more importantly, it is unclear how to use these techniques to prove any sort of superlinear lower bound on wire complexity, since there are functions that have threshold circuits with linearly many wires, but large communication complexity even after applying restrictions (take a generic read-once depth-$2$ majority formula for example).

Perhaps most closely related to our work is that of Gopalan and Servedio~\cite{GS} who use analytic techniques to prove correlation bounds against $\AC^0$ circuits augmented with a few threshold gates. Their idea is to use noise sensitivity bounds (as we do as well) to obtain correlation bounds for Parity against $\TAC^0$ circuits and then extend these results in the same way as in the
paper by  %
Podolskii~\cite{Podolskii} mentioned above. As a result, though, the result only yields non-trivial bounds when the number of threshold gates is bounded by $\log n$, whereas our result yields correlation bounds for up to $n^{1/2(d-1)}$ threshold gates.

\subsection{Proof techniques}

In recent years, there
has been much work  %
exploring the analytic properties (such as noise sensitivity) of linear threshold functions (LTFs) and their generalizations polynomial threshold functions (PTFs) (\eg, \cite{Ser,OS,DGJSV,HKM,DRST,MZ,Kane}). We show here that these techniques can be used in the context of constant-depth threshold circuits as well. In particular, using these techniques, we prove a qualitative refinement of Peres' theorem (see below) that may be of independent interest.

Our first result (\expref{Theorem}{thm:ns-gates} in \expref{Section}{sec:gates-basic}) is a tight correlation bound for Parity with threshold circuits of depth $d$ and gate complexity much smaller than $n^{1/2(d-1)}$. This generalizes both the results of Siu et al.~\cite{Siuetal}, who proved such a result for \emph{majority} circuits, and Impagliazzo, Paturi, and Saks~\cite{IPS}, who proved a worst-case lower bound of the same order. The proof uses a fundamental theorem of Peres~\cite{peres} on the noise sensitivity of LTFs; Peres' theorem has also been used by Klivans, O'Donnell, and Servedio~\cite{KOS} to obtain learning algorithms for functions of a few threshold gates. We use Peres' theorem to prove a noise sensitivity upper bound on small threshold circuits of constant depth. 

The observation underlying the proof is that the noise sensitivity of a function is exactly the expected variance of the function after applying a suitable random restriction (see also~\cite{ODonnell-hardness-amp}). Seen in this light, Peres' theorem says that, on application of a random restriction, any threshold function becomes quite biased in expectation and hence is well approximated by a constant function. Our analysis of the threshold circuit therefore proceeds by applying a random restriction to the circuit and replacing all the threshold gates at depth $d-1$ by the constants that best approximate them to obtain a circuit of depth $d-1$. A straightforward union bound tells us that the new circuit is a good approximation of the original circuit after the restriction. We repeat this procedure with the depth-$(d-1)$ circuit until the entire circuit becomes a constant, at which point we can say that after a suitable random restriction, the original circuit is well approximated by a constant, which means its variance is small. Hence, the noise sensitivity of the original circuit must be small as well and we are done.

This technique is expanded upon in \expref{Section}{sec:tac0r}, where we use a powerful noise sensitivity upper bound for low degree PTFs due to Kane~\cite{Kane} along with standard switching arguments~\cite{Hastad} to prove similar results for $\AC^0$ circuits augmented with almost $n^{1/2(d-1)}$ threshold gates. This yields \expref{Theorem}{MainThm3} by some standard results.

In \expref{Section}{sec:wires}, we consider the problem of extending the above correlation bounds to threshold circuits with small (slightly superlinear) wire complexity. The above proof breaks down even for depth-$2$ threshold circuits with a superlinear number of wires, since such circuits could have a superlinear number of gates and hence the union bound referred to above is no longer feasible.

In the case of depth-$2$ threshold circuits, we are nevertheless able to use Peres' theorem, along with ideas of~\cite{ABFR} to prove correlation bounds for Parity against circuits with nearly $n^{1.5}$
wires.  (This  %
result was independently obtained by
Kane and Williams in a recent paper~\cite{KW}.)  %
This result is tight, since by
a result of %
Siu et al.~\cite{Siuetal}, Parity can be well approximated by depth-$2$ circuits with $O(\sqrt{n})$ gates and hence $O(n^{1.5})$ wires. This argument is in \expref{Section}{sec:d2thr}.

Unfortunately, however, this technique
requires  %
us to set a large number of variables, which renders it unsuitable for larger depths. The reason for this is that, if we set a large number of variables to reduce the depth from some large constant $d$ to $d-1$, then we may be in a setting where the number of wires is much larger than the number of surviving variables and hence correlation bounds for Parity may no longer be possible at all.

We therefore use a different strategy to prove correlation bounds against larger constant depths. The linchpin in the argument is a qualitative refinement of Peres' theorem (\expref{Lemma}{lem:thresh-gate}) that says that on application of a random restriction to an LTF, with good probability, the variance of the LTF becomes \emph{negligible} (even exponentially small for suitable parameters). The proof of this argument is via anticoncentration results based on the Berry-Esseen theorem and the analysis of general threshold functions via a \emph{critical index} argument as in many recent
papers~\cite{Ser,OS,DGJSV,MZ}.   %

The above refinement of Peres' theorem allows us to proceed with our argument as in the gates case. We apply a random restriction to the circuit and by the refinement, with good probability (say $1-n^{-\Omega(1)}$) most gates end up exponentially close to constants. We can then set these ``imbalanced'' gates to constants and still apply a union bound to ensure that the new circuit is a good approximation to the old one. For the small number of gates that do not become imbalanced in this way, we set \emph{all} variables feeding into them. Since the number of such gates is small, we do not set too many variables. We now have a depth-$(d-1)$ circuit. Continuing in this way, we get a correlation bound of $n^{-\Omega(1)}$ with Parity. This gives part of \expref{Theorem}{MainThm1}.

We then strengthen this correlation bound to $\exp(-n^{\Omega(1)})$ for the \emph{Generalized Andreev function}, which, intuitively speaking, has the following property: even after applying any restriction that leaves a certain number of variables unfixed, the function has exponentially small correlation with any LTF on the surviving variables. To prove lower bounds for larger depth threshold circuits, we follow more or less the same strategy, except that in the above argument, we need most gates to become imbalanced with very high probability ($1-\exp(-n^{\Omega(1)})$). To ensure this, we use a \emph{bounded read Chernoff bound} due to Gavinsky, Lovett, Saks, and Srinivasan~\cite{GLSS}. We can use this technique to reduce depth as above as long as the number of threshold gates at depth $d-1$ is ``reasonably large.'' If the number of gates at depth $d-1$ is very small, then we simply guess the values of these few threshold gates and move them to the top of the circuit and proceed. This gives the other part of \expref{Theorem}{MainThm1}.

This latter depth-reduction result can be completely constructivized to design a satisfiability algorithm that runs in time $2^{n-n^{\Omega(1)}}$. The algorithm proceeds in the same way as the above argument, iteratively reducing the depth of the circuit. A subtlety arises when we replace imbalanced gates by constants, since we are changing the behaviour of the circuit on some (though very few) inputs. Thus, a circuit which was satisfiable only at one among these inputs might now end up unsatisfiable. However, we show that there is an efficient algorithm that enumerates these inputs and can hence check if there are satisfiable assignments to the circuits from among these inputs. 
This gives \expref{Theorem}{MainThm2}.

In \expref{Section}{sec:formulas}, we prove correlation bounds for the Generalized Andreev function against threshold \emph{formulas} of any arity and any depth. The proof is based on a retooling of the argument of Ne\v{c}iporuk for formulas of constant arity over any basis and yields a correlation bound as long as the wire complexity is at most $n^{2-\Omega(1)}$. 

\paragraph{Independent work of Kane and Williams~\cite{KW} and connections to our results.} Independent of our work, Kane and Williams~\cite{KW} proved some very interesting results on threshold circuit lower bounds. Their focus is on threshold circuits of depth $2$ and (a special case of) depth $3$ only, but in this regime they are able to show stronger superlinear gate lower bounds and superquadratic wire lower bounds on the complexity of an explicit function (which is closely related to the Generalized Andreev function referred to above). 

The techniques of~\cite{KW} are closely related to ours, since they also analyze the effect of random restrictions on threshold gates. The statement of the random restriction lemma in our
paper  %
is different from that in~\cite{KW}. While we are only able to prove that a threshold gate becomes highly biased with high probability under a random restriction,~\cite{KW} prove that with high probability, the threshold gate actually becomes a \emph{constant}. 

However, to obtain this stronger conclusion, the random restrictions from~\cite{KW} have to set \emph{most} variables, which makes it unsuitable for proving lower bounds for
larger   %
depths as mentioned above. This technique can recover our tight average-case gate lower bounds (\expref{Section}{sec:gates-basic}) and wire lower bounds (\expref{Section}{sec:d2thr}) for depth $2$.

While our random restriction lemmas are weaker (in that they don't make the threshold gate a constant), it makes sense to ask if we can use them to recover the threshold circuit lower bounds of~\cite{KW}. 

In \expref{Section}{sec:kw}, we show that this is true by using a natural strengthening of our main restriction lemma to prove a lower bound for threshold circuits that matches the lower bound of~\cite{KW} upto $n^{o(1)}$ factors. The results of this section were obtained subsequent to the work of Kane and Williams.

This result should not be considered an ``alternate proof,'' since we use many other ideas from~\cite{KW} in the proof (we have also used one of these ideas, namely \expref{Theorem}{thm:RSO94}, to improve our lower bounds for Threshold formulas from~\cite{CSS-conf} in \expref{Section}{sec:formulas}). Our interest is in understanding the relative power of the random restriction lemmas in the two results.

\section{Preliminaries}

Throughout, all logarithms will be taken to the base $2$.

\subsection{Basic Boolean function definitions}

A Boolean function on $n$ variables will be a function $f:\{-1,1\}^n\rightarrow\{-1,1\}$. We use the standard inner product on functions $f,g:\{-1,1\}^n \rightarrow \mathbb{R}$ defined by
\[
  \ip{f}{g} = \avg{x\sim \{-1,1\}^n}{f(x)g(x)}.
\]
(Unless specifically mentioned otherwise, we use $a\sim A$ for a finite set $A$ to denote that $a$ is a uniform random sample from the set $A$.)

Given Boolean functions $f,g$ on $n$ variables, the \emph{Correlation} between $f$ and $g$---denoted $\Corr(f,g)$---is defined as
\[
\Corr(f,g) := |\ip{f}{g}| = \left|\avg{x\sim \{-1,1\}^n}{f(x)g(x)}\right| = \bigl|2\prob{x}{f(x) = g(x)} - 1\bigr|.
\]
Also, we use $\delta(f,g)$ to denote the fractional distance between $f$ and $g$: \ie, $\delta(f,g) = \prob{x}{f(x)\neq g(x)}$. Note that for Boolean $f,g$, we have $\Corr(f,g) = |1-2\delta(f,g)|$. We say that $f$ is $\delta$-approximated by $g$ if $\delta(f,g) \leq \delta$.

We use $\Par_n$ to denote the parity function on $n$ variables. I.e., $\Par_n(x_1,\ldots,x_n) = \prod_{i=1}^n x_i$.

\begin{definition}[Restrictions]
A restriction on $n$ variables is a function $\rho:[n]\rightarrow \{-1,1,*\}$. A random restriction is a distribution over restrictions. We use $\mc{R}_p^n$ to denote the distribution over restrictions on $n$ variables obtained by setting each $\rho(x) = *$ with probability $p$ and to $1$ and $-1$ with probability ${(1-p)}/{2}$ each. We will often view the process of sampling a restriction from $\mc{R}_p^n$ as picking a pair $(I,y)$ where $I\subseteq [n]$ is obtained by picking each element of $[n]$ to be in $I$ with probability $p$ and $y\in \{-1,1\}^{n-|I|}$ is chosen uniformly at random.
\end{definition}

\begin{definition}[Restriction trees and Decision trees]
A \emph{restriction tree} $T$ on $\{-1,1\}^n$ of depth $h$ is a binary tree of depth $h$ all of whose internal nodes are labeled by one of $n$ variables, and the outgoing edges from an internal node are labeled $+1$ and $-1$; we assume that a node and its ancestor never query the same variable. Each leaf $\ell$ of $T$ defines a restriction $\rho_\ell$ that sets all the variables on the path from the root of the decision tree to $\ell$ and leaves the remaining variables unset. A random restriction tree $\mc{T}$ of depth $h$ is a distribution over restriction trees of depth $h$. 

Given a restriction tree $T$, the process of choosing a random edge out of each internal node generates a distribution over the leaves of the tree (note that this distribution is not uniform: the weight it puts on leaf $\ell$ at depth $d$ is $2^{-d}$). We use the notation $\ell\sim T$ to denote a leaf $\ell$ of $T$ picked according this distribution.

A \emph{decision tree} is a restriction tree all of whose leaves are labeled either by $+1$ or $-1$. We say a decision tree has size $s$ if the tree has $s$ leaves. We say a decision tree computes a function $f: \{-1,1\}^n \rightarrow \{-1,1\}$ if for each leaf $\ell$ of the tree, $f|_{\rho_\ell}$ is equal to the label of $\ell$.

\end{definition}

\begin{fact}[Facts about correlation]
\label{fac:corr}
Let $f,g,h:\{-1,1\}^n\rightarrow \{-1,1\}$ be arbitrary.
\begin{enumerate}
\item $\Corr(f,g)\in [0,1]$.
\item If $\Corr(f,g)\leq \varepsilon$ and $\delta(g,h) \leq \delta$, then $\Corr(f,h) \leq \varepsilon + 2\delta$.
\item Let $\mc{T}$ be any random restriction tree. Then
  \[
    \Corr(f,g)\leq \avg{T\sim \mc{T}, \ell\sim T}{\Corr(f|_{\rho_\ell},g|_{\rho_\ell})}.
    \]
\end{enumerate}
\end{fact}

\begin{definition}[Noise sensitivity and Variance~\cite{Odonnellbook}]
\label{def:ns}
Given a Boolean function $f:\{-1,1\}^n\rightarrow \{-1,1\}$ and a parameter $p \in [0,1]$, we define the \emph{noise sensitivity of $f$ with noise parameter $p$}---denoted $\NS_p(f)$---as follows. Pick $x\in \{-1,1\}^n$ uniformly at random and $y\in \{-1,1\}^n$ by negating (\ie, flipping) each bit of $x$ independently with probability $p$; we define
\[
  \NS_p(f) = \prob{(x,y)}{f(x)\neq f(y)}.
\]
The variance of $f$---denoted $\Var(f)$---is defined to be
$2\NS_{1/2}(f)$.
\end{definition}

\begin{fact}
\label{fac:var} 
Let $f:\{-1,1\}^n\rightarrow \{-1,1\}$ be any Boolean function. Let
\[
  p = \min\{\prob{x}{f(x)=1},\prob{x}{f(x) = -1}\}
\]
where $x$ is chosen uniformly from $\{-1,1\}^n$. Then, $\Var(f) = \Theta(p)$.
\end{fact}

\begin{proposition}
\label{prop:ns} 
Let $f:\{-1,1\}^n\rightarrow \{-1,1\}$ be any Boolean function. Then,
\begin{enumerate}
\item for $p\leq 1/2$, $\NS_p(f) = ({1}/{2})\avg{\rho\sim \mc{R}_{2p}^n}{\Var(f|_\rho)}$;
\item if $p\geq {1}/{n}$, then $\Corr(f,\Par_n)\leq O(\NS_p(f))$.
\end{enumerate}
\end{proposition}

The above proposition seems to be folklore, but we couldn't find explicit proofs in the literature. For completeness we present the proof below.

\begin{proof}
  For point 1, we know that
  \[
    \NS_p(f) = \prob{(x,y)}{f(x)\neq f(y)}
  \]
  where $x$ and $y$ are sampled as in
  \expref{Definition}{def:ns}. Alternately, we may also think of
  sampling $(x,y)$ in the following way: choose
  $\rho = (I,z)\sim \mc{R}_{2p}^n$ and for the locations indexed by
  $I$ we choose $x',y'\in \{-1,1\}^{|I|}$ independently and uniformly
  at random to define strings $x$ and $y$ respectively. Hence, we have
\[
\NS_p(f) = \prob{x,y}{f(x)\neq f(y)} = \avg{\rho}{\prob{x',y'}{f|_\rho(x')\neq f|_{\rho}(y')}} = \avg{\rho}{\frac{1}{2}\Var(f|_\rho)}.
\]

We now proceed with point 2. As $\NS_p(f)$ is a decreasing function of $p$~\cite{Odonnellbook}, we may assume that $p = {1}/{n}\leq {1}/{2}$ and hence we have
\[
  \NS_{1/n}(f) = \frac{1}{2}\avg{\rho\sim \mc{R}_{2/n}^n}{\Var(f|_\rho)}.
\]
Note that for $\rho = (I,y)$ chosen as above, the probability that
$I\neq \emptyset$ is $\Omega(1)$. Hence we have
\[
\NS'_{1/n}(f) := \frac{1}{2}\avg{\rho\sim \mc{R}_{2/n}^n}{\Var(f|_\rho)\ |\ I\neq \emptyset} \leq \frac{\NS_{1/n}(f)}{\prob{I}{I\neq \emptyset}} =  O(\NS_{1/n}(f)).
\]
Further, note that for any $m\geq 1$ and any Boolean function $g:\{-1,1\}^m\rightarrow \{-1,1\}$, its distance from either the constant function $1$ or the constant function $-1$ is at most $\Var(g)/2$. Since $\Par_m$ has correlation $0$ with any constant function, using \expref{Fact}{fac:corr}, we have $\Corr(\Par_m,g)\leq \Var(g)/2$. 

Using \expref{Fact}{fac:corr} again, we get
\begin{align*}
\Corr(\Par_n,f) &\leq \avg{\rho\sim \mc{R}_{2/n}^n}{\Corr(\Par_n|_\rho,f|_\rho)\ |\ I\neq \emptyset} = \avg{\rho\sim \mc{R}_{2/n}^n}{\Corr(\Par_{|I|},f|_\rho)\ |\ I\neq \emptyset}\\
&\leq \avg{\rho\sim \mc{R}_{2/n}^n}{\frac{1}{2}\Var(g)\ |\ I\neq \emptyset} = \NS'_{1/n}(f) = O(\NS_{1/n}(f)).\qedhere
\end{align*}
\end{proof}

\subsection{Threshold functions and circuits}
\label{sec:defnsTF}

\begin{definition}[Threshold functions and gates]
A \emph{Threshold gate} is a gate $\phi$ labeled with a pair $(w,\theta)$ where $w\in \mathbb{R}^m$ for some $m\in \mathbb{N}$ and $\theta\in \mathbb{R}$. The gate computes the Boolean function $f_\phi:\{-1,1\}^m\rightarrow \{-1,1\}$ defined by $f_\phi(x) = \sgn(\ip{w}{x}-\theta)$ (we define $\sgn(0) = -1$ for the sake of this definition). The fan-in of the gate $\phi$---denoted $\fanin(\phi)$---is $m$. A \emph{Linear Threshold function} (LTF) is a Boolean function that can be represented by a Threshold gate. More generally, a Boolean function $f:\{-1,1\}^n\rightarrow \{-1,1\}$ is a degree-$D$ \emph{Polynomial Threshold function} (PTF) if there is a degree-$D$ polynomial $p(x)$ such that $f(x) = \sgn(p(x))$ for all $x\in \{-1,1\}^n$.
\end{definition}

\begin{definition}[Threshold circuits]
A \emph{Threshold circuit $C$} is a Boolean circuit whose gates are all threshold gates. There are designated output gates, which compute the functions computed by the circuit. Unless explicitly mentioned, however, we assume that our threshold circuits have a unique output gate. The \emph{gate complexity} of $C$ is the number of (non-input) gates in the circuit, while the \emph{wire complexity} is the sum of all the fan-ins of the various gates.

A \emph{Threshold map} from $n$ to $m$ variables is a depth-$1$ threshold circuit $C$ with $n$ inputs and  $m$ outputs. We say that such a map is \emph{read-$k$} if each input variable is an input to at most $k$ of the threshold gates in $C$.
\end{definition}

The following theorem bounds the number of Threshold functions of a few other arbitrary functions on $n$ bits.

\begin{theorem}[\cite{RSO94}]\label{thm:RSO94}
Let $f_1,\ldots, f_s$ be a collection of Boolean functions on $n$ variables.
Then there are at most $2^{O(ns)}$ distinct functions $g$ of the form
\[
g(x_1,\ldots, x_n) = \sgn\left(  \sum_{i=1}^{s} w_i f_i(x_1,\ldots,x_n) - \theta \right),
\] 
where $w_1,\ldots, w_s,\theta \in \mathbb{R}$. 
\end{theorem}

Applying the above theorem in the case that $s=n$ and the $f_i$s are just the input bits, we obtain the following.

\begin{corollary}
\label{cor:numthr}
The number of distinct LTFs on $n$ bits is at most $2^{O(n^2)}$.
\end{corollary}

\begin{definition}[Restrictions of threshold gates and circuits]
\label{def:restrict}
Given a threshold gate $\phi$ of fan-in $m$ labeled by the pair $(w,\theta)$ and a restriction $\rho$ on $m$ variables, we use $\phi|_\rho$ to denote the threshold gate over the variables indexed by $\rho^{-1}(*)$ obtained in the natural way by setting variables according to $\rho$.
\end{definition}

We will also need Peres' theorem, which bounds the noise sensitivity of threshold functions.

\begin{theorem}[Peres' theorem~\cite{peres,Odonnellbook}]
\label{thm:Peres}
Let $f:\{-1,1\}^n\rightarrow \{-1,1\}$ be any LTF\@. Then, 
\[
\avg{\rho\sim \mc{R}_p^n}{\mathrm{Var}(f|_\rho)} = \NS_{\frac{p}{2}}(f) = O(\sqrt{p}).
\]
\end{theorem}

Using the above for $p = 1/n$ and \expref{Proposition}{prop:ns}, we obtain

\begin{corollary}
\label{cor:Peres}
Let $f:\{-1,1\}^n\rightarrow\{-1,1\}$ be any threshold function. Then $\Corr(f,\Par_n) \leq O(n^{-1/2})$.
\end{corollary}

\subsection{Description lengths and Kolmogorov complexity}

\begin{definition}[Kolmogorov Complexity]
The Kolmogorov complexity of an $n$-bit Boolean string $x$ is the length of the shortest bit string of the form $(M,w)$ where $M$ is the description of a Turing Machine and $w$ an input to $M$ such that $M(w) = x$. We use $K(x)$ to denote the Kolmogorov complexity of $x$.
\end{definition}

\begin{fact}
\label{fac:kol-lbd}
For any $\alpha \in (0,1)$, the fraction of $n$-bit strings $x$ satisfying $K(x) \leq (1-\alpha)n$ is at most $2^{-\alpha n+1}$.
\end{fact}

\subsection{The Generalized Andreev function}
\label{sec:andreev}
We state here the definition of a generalization of Andreev's function, due to Komargodski and Raz, and Chen, Kabanets, Kolokolova, Shaltiel, and Zuckerman~\cite{KR,CKKSZ}. This function will be used to give strong correlation bounds against constant-depth threshold circuits with slightly superlinear wire complexity.

We first need some definitions.

\begin{definition}[Bit-fixing extractor]
A \emph{subcube} $C\subseteq \{-1,1\}^n$ is a set obtained by restricting some subset of the $n$ Boolean variables to some fixed values. A function $E:\{-1,1\}^n\rightarrow \{-1,1\}^m$ is a \emph{$(n,k,m,\zeta)$ bit-fixing extractor} if for every random variable $X$ that is uniform on a subcube of $\{-1,1\}^n$ of dimension at least $k$, the function $E(X)$ is $\zeta$-close to uniform on $\{-1,1\}^m$.
\end{definition}

We have the following explicit construction of a bit-fixing extractor.

\begin{theorem}[\cite{Rao}]
\label{thm:rao}
There is an absolute constant $c\geq 1$ so that the following holds. There is a polynomial-time computable function $E:\{-1,1\}^n\rightarrow \{-1,1\}^m$ that is an $(n,k,m,\zeta)$-bit fixing extractor for any $k \geq (\log n)^c$, $m = 0.9k$, and $\zeta \leq 2^{-k^{\Omega(1)}}$.
\end{theorem}

\newcommand{\Enc}{\mathrm{Enc}}
Also recall that a function $\Enc:\{-1,1\}^a \rightarrow \{-1,1\}^b$ defines \emph{$(\alpha,L)$-error-correcting code} for parameters $\alpha\in [0,1]$ and $L\in \mathbb{N}$ if for any $z\in \{-1,1\}^b$, the number of elements in the image of $\Enc$ that are at relative Hamming distance at most $\alpha$ from $z$ is bounded by $L$.

The following theorem is
folklore   %
and
is  %
stated explicitly in the
paper by   %
Chen et al.~\cite{CKKSZ}.

\begin{theorem}[\cite{CKKSZ}, Theorem 6.4]
\label{thm:ckksz-codes}
Let $r = n^{\beta}$ for any fixed $0< \beta < 1$. There exists an $(\alpha, L)$-error correcting code with $\Enc:\{-1,1\}^{4n}\rightarrow \{-1,1\}^{2^r}$ where $\alpha = {1}/{2}-O(2^{-r/4})$ and $L = O(2^{r/2})$. Further, there is a $\poly(n)$ time algorithm, which when given as input $x\in \{-1,1\}^{4n}$ and $i\in [2^r]$ in binary, outputs $\Enc(x)_i$, the $i$th bit of $\Enc(x)$.
\end{theorem}

Now we can define the Generalized Andreev function as in~\cite{CKKSZ}. The function is $F:\{-1,1\}^{5n}\times \{-1,1\}^{n}\rightarrow \{-1,1\}$ and is defined as follows. Let $\gamma > 0$ be a constant parameter. The parameter will be fixed later according to the application at hand.

Let $E$ be \emph{any} $(n,n^\gamma,m = 0.9 n^\gamma,2^{-n^{\Omega(\gamma)}})$ extractor (we can obtain an explicit one using \expref{Theorem}{thm:rao}). We interpret the output of $E$ as an integer from $[2^m]$ in the natural way. Let $\Enc:\{-1,1\}^{4n}\rightarrow \{-1,1\}^{2^m}$ define a $({1}/{2}-O(2^{-m/4}),2^{m/2})$-error correcting code as in \expref{Theorem}{thm:ckksz-codes}. Then, we define $F(x_1,x_2)$ by
\begin{equation}
\label{eq:defnF}
F(x_1,x_2) = \Enc(x_1)_{E(x_2)}.
\end{equation}

Given $a \in \{-1,1\}^{4n}$, we use $F_a(\cdot)$ to denote the resulting sub-function on $n$ bits obtained by fixing $x_1=a$. 

The following lemma was proved as part of Theorem 6.5 in~\cite{CKKSZ}.

\begin{lemma}[\cite{CKKSZ}, Theorem 6.5]\label{lem:ckksz-corr}
Let $C$ be any circuit on $n^\gamma$ variables with binary description length $\sigma(C) \leq n$ according to some fixed encoding scheme. Let $\rho$ be any restriction of $n$ variables leaving $n^{\gamma}$ variables unfixed.
Let $f(y)  := F_a|_\rho(y)$ for $a\in \{-1,1\}^{4n}$ satisfying $K(a) \geq 3n$. Then
\[
	\Corr(f, C)  \leq \exp(-n^{\Omega(\gamma)}).
\]
\end{lemma}
\subsection{Concentration bounds}

We state a collection of concentration bounds that we will need in our proofs. The proofs of Theorems~\ref{thm:chernoff} and~\ref{thm:chernoff-mul} may be found in the excellent book by Dubhashi and Panconesi~\cite{DP}.

\begin{theorem}[Chernoff bound]
\label{thm:chernoff}
Let $w\in \mathbb{R}^n$ be arbitrary and $x$ is chosen uniformly from $\{-1,1\}^n$. Then
\[
\prob{x}{|\ip{w}{x}| \geq t\cdot\mynorm{w}_2}\leq \exp(-\Omega(t^2)).
\]
\end{theorem}

\begin{definition}[Imbalance]
We say that a threshold gate $\phi$ labeled by $(w,\theta)$ is $t$-imbalanced if $|\theta|\geq t\cdot \mynorm{w}_2$ and $t$-balanced otherwise.
\end{definition}

We also need a multiplicative form of the Chernoff bound for sums of Boolean random variables. 

\begin{theorem}[Multiplicative Chernoff bound]
\label{thm:chernoff-mul}
Let $Y_1,\ldots,Y_m$ be independent Boolean random variables such that $\avg{}{Y_i} = p_i$ for each $i\in [m]$. Let $p$ denote the average of the $p_i$. Then, for any $\varepsilon >0$
\[
\prob{}{|\sum_i Y_i - pm| \geq \varepsilon pm}\leq \exp(-\Omega(\varepsilon^2 pm)).
\]
\end{theorem}

We will need the following variant of the Chernoff bound that may be found in the survey of Chung and Lu~\cite{ChungLu}. 

\begin{theorem}[Theorem 3.4 in~\cite{ChungLu}]
\label{thm:chunglu}
Fix any parameter $p\in (0,1)$. Let $X_1,\ldots,X_n$ be independent Boolean valued random variables such $\avg{}{X_i} = p$ for each $i\in [n]$. Let $X = \sum_{i\in [n]} a_i X_i$ where each $a_i\geq 0$. Then, 
\[
\prob{}{X\geq \avg{}{X} + \lambda} \leq \exp(-\lambda^2/(2(\nu + a\lambda))),
\]
where $a = \max\{a_1,\ldots,a_n\}$ and $\nu = p\cdot \left(\sum_i a_i^2\right)$. In the case that $\lambda \geq \avg{}{X}$, we have $\nu \leq pa\cdot \sum_i a_i = a\avg{}{X}\leq a\lambda$ and hence we obtain
\[
\prob{}{X\geq \avg{}{X} + \lambda} \leq \exp(-\lambda/(4a)).
\]
\end{theorem}

Let $Y_1,\ldots,Y_m$ be random variables defined as functions of independent random variables $X_1,\ldots,X_n$. For $i\in [m]$, let $S_i\subseteq [n]$ index those random variables among $X_1,\ldots,X_n$ that influence $Y_i$. We say that $Y_1,\ldots,Y_m$ are read-$k$ random variables if any $j\in [n]$ belongs to $S_i$ for at most $k$ different $i\in [m]$.

The notation $D(p||q)$ represents the KL-divergence (see, \eg,~\cite{CT}) between the two probability distributions on $\{0,1\}$ where the probabilities assigned to $1$ are $p$ and $q$ respectively.

\begin{theorem}[A read-$k$ Chernoff bound~\cite{GLSS}]
\label{thm:GLSS}
Let $Y_1,\ldots,Y_m$ be $\{0,1\}$-valued read-$k$ random variables such that $\avg{}{Y_i} = p_i$. Let $p$ denote the average of $p_1,\ldots,p_m$. Then, for any $\varepsilon > 0$,
\[
\prob{}{\sum_i Y_i \geq pm (1+\varepsilon)} \leq \exp\bigl(-D(p(1+\varepsilon)\;\|\;p)m/k\bigr).
\]
\end{theorem}

It was pointed out to us by an anonymous reviewer that the above actually follows easily from an older, stronger result of Janson~\cite{Janson}. We cite the result from~\cite{GLSS} since it is stated in the form above, which is useful for us.

Using standard estimates on the KL-divergence, \expref{Theorem}{thm:GLSS} implies the following.

\begin{corollary}
\label{cor:GLSS}
Let $Y_1,\ldots,Y_m$ be as in the statement of \expref{Theorem}{thm:GLSS} and assume $\avg{}{\sum_i Y_i}\leq \mu$. Then,
\[
\prob{}{\sum_i Y_i \geq 2\mu} \leq \exp(-\Omega(\mu/k)).
\]

\end{corollary}

\section{Correlation bounds against threshold circuits with small gate complexity}
\label{sec:gates-basic}

This section serves as a warm up to our main results in the following sections. Here we use a simple version of our technique to show that constant-depth threshold circuits with a small number of gates cannot correlate well with the Parity function. The following is a consequence of our main result.

\begin{corollary}
\label{cor:gates}
Fix any $d \geq 2$. Assume that $C$ is a depth-$d$ threshold circuit over $n$ variables with $k \leq n^{{1}/{(2(d-1))}}$ threshold gates and let $\delta = k/n^{{1}/{(2(d-1))}}$. Then, $\Corr(C,\Par_n) \leq O(\delta^{(1-{1}/{d})})$. In particular, for any constant $d$, $\Corr(C,\Par_n) = o(1)$ unless $k = \Omega(n^{1/2(d-1)})$.
\end{corollary}

It should be noted that Nisan~\cite{Nisan-threshold} already proved stronger correlation bounds for the \emph{Inner Product} function against \emph{any} threshold circuit (not necessarily constant-depth) with a sub-linear (much smaller than $n/\log n$) number of threshold gates. The idea of his proof is to first show that each threshold gate on $n$ variables has a $\delta$-error randomized communication protocol with complexity $O(\log(n/\delta))$~\cite[Theorem 1]{Nisan-threshold}. One can use this to show that any threshold circuit as in the theorem can be written as a decision tree of depth $n/k$ querying threshold functions and hence has an $\exp(-\Omega(k))$-error protocol of complexity at most $n/10$. Standard results in communication complexity imply that any such function can have correlation at most $\exp(-\Omega(k))$ with inner product.

However, such techniques cannot be used to obtain lower bounds or correlation bounds for the parity function, since the parity function has low communication complexity, even in the deterministic setting. An even bigger disadvantage to this technique is that it cannot be used to obtain \emph{any} superlinear lower bound on the wire complexity, since threshold circuits with a linear number of wires can easily compute functions with high communication complexity, such as the Disjointness function.  

The techniques we use here can be used to give nearly tight~\cite{Siuetal} correlation bounds for the parity function (and can also be strengthened to the setting of small wire complexity, as we will show later). In fact, we prove something stronger: we upper bound the noise sensitivity of small constant-depth threshold circuits, which additionally implies the existence of non-trivial learning algorithms~\cite{KOS,GS} for such circuits. Further, our techniques also imply noise sensitivity bounds for $\AC^0$ circuits augmented with a small number of threshold gates. For the sake of exposition, we postpone these generalizations to \expref{Section}{sec:tac0r} and prove only the noise sensitivity result for constant-depth threshold circuits in this section.

\subsection{Correlation bounds via noise sensitivity}

The main result of this section is the following.

\begin{theorem}
\label{thm:ns-gates}
Let $C$ be a depth-$d$ threshold circuit with at most $k$ threshold gates. Then, for any parameters $p,q\in [0,1]$, we have
\[
\NS_{p^{d-1}q}(C) \leq O(k\sqrt{p} + \sqrt{q}).
\]
\end{theorem}

\begin{proof}
We assume that $q \leq {1}/{2}$, since otherwise we always have $\NS_{p^{d-1}q}(C)\leq 1\leq 2\sqrt{q}$. 

We prove below that for $p_d:= 2p^{d-1}q\in [0,1]$ and $\rho_d \sim \mc{R}_{p_d}^n$ ($n$ is the number of input variables to $C$), we have
\begin{equation}
\label{eq:ns1}
\avg{\rho_d}{\Var(C|_{\rho_d})} \leq O(k\sqrt{p} + \sqrt{q}).
\end{equation}

This will imply the theorem, since by \expref{Proposition}{prop:ns}, we have
\[
  \NS_{p^{d-1}q}(C) = \frac{1}{2}\avg{\rho_d}{\Var(C|_{\rho_d})}\,.
\]

The proof of (\ref{eq:ns1}) is by induction on the depth $d$ of the circuit. The base case $d=1$ is just Peres' theorem (\expref{Theorem}{thm:Peres}). 

Now assume that $C$ has depth $d > 1$. Let $k_1$ be the number of threshold circuits at depth $d-1$ in the circuit. We choose a random restriction $\rho\sim \mc{R}_p^n$ and consider the circuit $C|_{\rho}$. It is easy to check that 
\begin{equation}
\label{eq:ns2}
\avg{\rho_d}{\Var(C|_{\rho_d})} = \avg{\rho}{\avg{\rho_{d-1}}{\Var((C|_{\rho})|_{\rho_{d-1}})}}
\end{equation}
and hence to prove (\ref{eq:ns1}), it suffices to bound the expectation of $\Var((C|_{\rho})|_{\rho_{d-1}})$. 

Let us first consider the circuit $C|_\rho$. Peres' theorem tells us that on application of the restriction $\rho$, each threshold gate at depth $d-1$ becomes quite biased on average. Formally, by \expref{Theorem}{thm:Peres} and \expref{Fact}{fac:var}, for each threshold gate $\phi$ at depth $d-1$, there is a bit $b_{\phi,\rho}\in \{-1,1\}$ such that 
\[
\avg{\rho}{\prob{x\in \{-1,1\}^{|\rho^{-1}(*)|}}{\phi|_\rho(x) \neq b_{\phi,\rho}}}\leq O(\sqrt{p}).
\]

In particular, replacing $\phi|_\rho$ by $b_{\phi,\rho}$ in the circuit $C|_\rho$ yields a circuit that differs from $C|_\rho$ on only an $O(\sqrt{p})$ fraction of inputs (in expectation). Applying this replacement to each of the $k_1$ threshold gates at depth $d-1$ yields a circuit $C'_\rho$ with $k-k_1$ threshold gates and depth $d-1$ such that
\begin{equation}
\label{eq:ns3}
\avg{\rho}{\delta(C|_\rho,C'_\rho)} \leq O(k_1\sqrt{p}),
\end{equation}
where $\delta(C|_\rho,C'_\rho)$ denotes the fraction of inputs on which the two circuits differ. On the other hand, we can apply the inductive hypothesis to $C'_\rho$ to obtain
\begin{equation}
\label{eq:ns4}
\avg{\rho_{d-1}}{\Var((C'_\rho)|_{\rho_{d-1}})} \leq O((k-k_1)\sqrt{p} + \sqrt{q}).
\end{equation}
Therefore, to infer (\ref{eq:ns1}), we put the above together with (\ref{eq:ns3}) and the following elementary fact.

\begin{proposition}
\label{prop:var}
Say $f,g:\{-1,1\}^m\rightarrow \{-1,1\}$ and $\delta = \delta(f,g)$. Then, for any $r\in [0,1]$, we have
\[
  \avg{\rho\sim \mc{R}_r^n}{\Var(f|_\rho)} \leq \avg{\rho\sim \mc{R}_r^n}{\Var(g|_\rho)} +  4\delta\,.
\]
\end{proposition}
\begin{proof}[Proof of \expref{Proposition}{prop:var}]
  By \expref{Proposition}{prop:ns}, we know that
  \[
    \avg{\rho\sim \mc{R}_r^n}{\Var(f|_\rho)} = 2\NS_{r/2}(f)\,,
  \]
  and similarly for $g$. By definition of noise sensitivity, we have
  \[
    \NS_{r/2}(f) = \prob{(x,y)}{f(x) \neq f(y)}
  \]
  where
  $x\in \{-1,1\}^m$ is chosen uniformly at random and $y$ is chosen by
  flipping each bit of $x$ with probability $r/2$. Note that each of
  $x$ and $y$ is individually uniformly distributed over $\{-1,1\}^m$
  and hence, both $f(x) = g(x)$ and $f(y) = g(y)$ hold with
  probability at least $1-2\delta$. This yields
\[
\NS_{r/2}(f) = \prob{(x,y)}{f(x) \neq f(y)} \leq \prob{(x,y)}{g(x) \neq g(y)} + 2\delta = \NS_{r/2}(g) + 2\delta,
\]
which implies the claimed bound.
\end{proof}
\end{proof}

The above theorem yields the correlation bound for Parity stated above (\expref{Corollary}{cor:gates}) as we now show.

\begin{proof}[Proof of \expref{Corollary}{cor:gates}]
  We apply \expref{Theorem}{thm:ns-gates} with the following optimized parameters:
  \[
    p = \frac{1}{n^{1/d}}\cdot \frac{1}{k^{2/d}}
  \]
  and $q\in[0,1]$ such that $p^{d-1}q = {1}/{n}$. It may be
  verified that for this setting of parameters,
  \expref{Theorem}{thm:ns-gates} gives us
\[
\NS_{1/n}(C) \leq O\left(\frac{k^{1-1/d}}{n^{1/(2d)}}\right)\leq O(\delta^{1-\frac{1}{d}}).
\] 
As noted in \expref{Proposition}{prop:ns}, we have $\Corr(C,\Par_n)\leq O(\NS_{1/n}(C))$. This completes the proof.
\end{proof}

\begin{remark}
\label{rem:GS}
It is instructive to compare the above technique with the closely related work of Gopalan and Servedio~\cite{GS}. The techniques of~\cite{GS} applied to the setting of \expref{Theorem}{thm:ns-gates} show that $\NS_p(C) \leq O(k2^k\sqrt{p})$, which gives a better dependence on the noise parameter $p$, but a much worse dependence on $k$. Indeed, this is not surprising since in this setting, the technique of Gopalan and Servedio does not use the fact that the circuit is of depth $d$. The threshold circuit is converted to a decision tree of depth $k$ querying threshold functions and it is this tree that is analyzed. 

We believe that the right answer should incorporate the best of both bounds:
\[
  \NS_p(f) \leq O_d(k^{d-1}\cdot \sqrt{p}).
\]
As in \expref{Corollary}{cor:gates}, this would show that
$\Corr(C,\Par_n) = o(1)$ if $k = o(n^{1/2(d-1)})$, but additionally,
we would also get $\Corr(C,\Par_n) \leq n^{-{1}/{2}+o(1)}$ as long as
$k = n^{o(1)}$ and $d$ is a constant.
\end{remark}

It is known from the work of Siu, Roychowdhury and Kailath~\cite[Theorem 7]{Siuetal} that \expref{Corollary}{cor:gates} is tight in the sense that there do exist circuits of gate complexity roughly $n^{1/2(d-1)}$ that have significant correlation with $\Par_n$. More formally,

\begin{theorem}[Theorem 7 in~\cite{Siuetal}]
\label{thm:ubd-gates}
Let $\varepsilon > 0$ be an arbitrary constant. Then, there is a threshold circuit of depth $d$ with $O(d)\cdot (n\log(1/\varepsilon))^{1/2(d-1)}$ gates that computes $\Par_n$ correctly on a $1-\varepsilon$ fraction of inputs.
\end{theorem}

\section{Correlation bounds against threshold circuits with small wire complexity}
\label{sec:wires}

In this section, we prove correlation bounds against threshold circuits of small \emph{wire complexity.} We prove three results of this kind. 

The first result, proved in \expref{Section}{sec:d2thr}, is a near-optimal correlation bound against depth-$2$ circuits computing the Parity function. This result was obtained independently by Kane and Williams~\cite{KW}.

\begin{theorem}[Correlation bounds against depth-$2$ threshold circuits]
\label{thm:d2thr}
Fix any constant $\varepsilon < {1}/{2}$. Let $\gamma = ({1}/{2})-\varepsilon$. Any depth-2 threshold circuit on $n$ variables with at most $n^{1+\varepsilon}$ wires has correlation at most $n^{-\Omega(\gamma)}$ with the Parity function on $n$ variables.
\end{theorem}

Note that the above theorem is nearly tight, since by \expref{Theorem}{thm:ubd-gates}, there is a depth-$2$ circuit with $O(\sqrt{n})$ gates (and hence $O(n^{3/2})$ wires) that computes Parity on $n$ variables correctly with high probability. The proof is a simple argument based on Peres' theorem along with a correlation bound against low-degree polynomial threshold functions due to Aspnes, Beigel, Furst and Rudich~\cite{ABFR}. 

The advantage of the above technique is that it yields a near-optimal lower bound in the depth-$2$ case. Unfortunately, however, it does not directly extend to depths larger than $2$. The main results of the section, stated below, show how to obtain correlation bounds against all constant depths.

The second result, proved in \expref{Section}{subsec:corr-par} via a more involved argument, yields a correlation bound for the parity function against all constant depths. 

\begin{theorem}[Correlation bounds for parity]
\label{thm:corr-parity}
For any $d \geq 1$, there is an $\varepsilon_d = 1/2^{O(d)}$ such that any  depth-$d$ threshold circuit $C$ with at most $n^{1+\varepsilon_d}$ wires satisfies $\Corr(C,\Par_n)\leq O(n^{-\varepsilon_d})$ where the $O(\cdot)$ hides absolute constants (independent of $d$ and $n$).
\end{theorem}

For comparison, we note that a result of Paturi and Saks~\cite{PS} shows that the Parity function can be computed by threshold circuits of depth $d$ with $n^{1+2^{-\Omega(d)}}$ many wires. A \emph{worst-case} lower bound of a similar form was proved by Impagliazzo et al.~\cite{IPS}.

Finally, we are also able to extend the techniques used in the proof of the above theorem to prove exponentially small correlation bounds. This kind of correlation bound cannot be proved for the Parity function since the Parity function on $n$ variables has correlation $\Omega(1/\sqrt{n})$ with the Majority function on $n$ variables (see, \eg,~\cite[Section 5.3]{Odonnellbook}), which clearly has a threshold circuit with only $n$ wires. 

We prove such a correlation bound in \expref{Section}{subsec:corr-andreev} for the Generalized Andreev function from \expref{Section}{sec:andreev}. For some technical reasons, we prove a slightly stronger result (\expref{Theorem}{thm:genstrong}) from which we obtain the following consequence. 

\begin{corollary}[Correlation bounds for Andreev's function]
\label{cor:corr-andreev}
For any constant $d \geq 1$, there is an $\varepsilon_d = {1}/{2^{O(d)}}$ such that the following holds. Let $F$ be the Generalized Andreev function on $5n$ variables as defined in \expref{Section}{sec:andreev} for any constant $\gamma < 1/6$. Any  depth-$d$ threshold circuit $C$ of wire complexity at most $n^{1+\varepsilon_d}$ satisfies $\Corr(C,F)\leq \exp(-n^{\Omega(\varepsilon_d)})$ where the $\Omega(\cdot)$ hides constants independent of $d$ and $n$.
\end{corollary}

We now state a key lemma that will be used in the proofs of our correlation bounds in Sections~\ref{subsec:corr-par} and~\ref{subsec:corr-andreev}. The lemma will be proved in \expref{Section}{subsec:msl}. 

Recall that a threshold gate with label $(w,\theta)$ is $t$-balanced if $|\theta|\leq t\cdot \mynorm{w}_2$.

\begin{lemma}[Main Structural lemma for threshold gates]
\label{lem:thresh-gate}
The following holds for some absolute constant $p_0\in [0,1]$. For any threshold gate $\phi$ over $n$ variables with label $(w,\theta)$ and any $p\in [0,p_0]$, we have\footnote{Recall from \expref{Definition}{def:restrict} that $\phi|_\rho$ is the threshold gate obtained by setting variables outside $\rho^{-1}(*)$ according to $\rho$.}
\[
\prob{\rho\sim \mc{R}^n_p}{\text{$\phi|_\rho$ is $\frac{1}{p^{\Omega(1)}}$-balanced}} \leq p^{\Omega(1)}.
\]
\end{lemma}

\subsection{Correlation bounds against depth-2 threshold circuits computing Parity}
\label{sec:d2thr}

In this section, we prove \expref{Theorem}{thm:d2thr}. The proof is based on the following two subclaims. 

\begin{theorem}[Aspnes, Beigel, Furst, and Rudich~\cite{ABFR}]
\label{thm:ABFR}
Any degree-$t$ polynomial threshold function (PTF)\footnote{We refer the reader to \expref{Section}{sec:defnsTF} for the definition of PTFs.} has correlation at most $O(t/\sqrt{m})$ with the parity function on $m$ variables.
\end{theorem}

\begin{lemma}
\label{lem:d2depthredn}
Let $\varepsilon,\gamma$ be as in the statement of \expref{Theorem}{thm:d2thr} and let $\alpha$ denote $\gamma/3$. Say $C$ denotes a depth-$2$ threshold circuit of wire complexity $n^{1+\varepsilon}$ and let $f_1,\ldots,f_t$ be the LTFs computed by $C$ at depth-$1$. Under a random restriction $\rho$ with $*$-probability $p = {1}/{n^{1-\alpha}}$, with probability at least $1-n^{-\Omega(\gamma)}$, the circuit $C|_\rho$ is  $n^{-\Omega(\gamma)}$-approximated by a circuit $\tilde{C}_\rho$ which is obtained from $C$ by replacing each of the $f_i|_\rho$s by an $O(n^{\alpha/2-\Omega(\gamma)})$-junta $g_i$.
\end{lemma}

Assuming the above two claims, we can finish the proof of \expref{Theorem}{thm:d2thr} easily as follows. 

Let $C$ be a circuit of wire complexity $n^{1+\varepsilon}$. We apply a random restriction $\rho$ with $*$-probability $p = {1}/{n^{1-\alpha}}$ as in \expref{Lemma}{lem:d2depthredn}. Call the restriction \emph{good} if there is a circuit $\tilde{C}_\rho$ as in the lemma that $n^{-\Omega(\gamma)}$-approximates $C|_\rho$ and $|\rho^{-1}(*)|\geq n^{\alpha}/2$. The probability that the first of these events does not occur is at most $n^{-\Omega(\gamma)}$ by \expref{Lemma}{lem:d2depthredn} and the probability of the second is at most $\exp(\Omega(n^{-\alpha}))$ by a Chernoff bound (\expref{Theorem}{thm:chernoff-mul}). Thus, the probability that $\rho$ is not good is at most
\[
  n^{-\Omega(\gamma)} + \exp(-\Omega(n^{-\alpha})) \leq n^{-\Omega(\gamma)}.
\]

Say $\rho$ is a good restriction. Note that each restricted LTF at depth-$1$ in the circuit $\tilde{C}_\rho$  is an $O(n^{\alpha/2-\Omega(\gamma)})$-junta and hence can be represented \emph{exactly} by a polynomial of the same degree. This implies that $\tilde{C}_\rho$ is a $O(n^{\alpha/2-\Omega(\gamma)})$-degree PTF and hence, by \expref{Theorem}{thm:ABFR}, has correlation at most $n^{-\Omega(\gamma)}$ with the Parity function (on the remaining $n^{\alpha}/2$ variables). Moreover, then $C|_\rho$ is well-approximated by $\tilde{C}_\rho$ and hence has correlation at most $n^{-\Omega(\gamma)} + n^{-\Omega(\gamma)}$ with parity. 

Upper bounding the correlation by $1$ for bad restrictions, we see that the overall correlation is at most $n^{-\Omega(\gamma)}$.

We now prove \expref{Lemma}{lem:d2depthredn}.

\begin{proof}[Proof of \expref{Lemma}{lem:d2depthredn}]
Let $f_1,\ldots,f_t$ be the LTFs appearing at depth $1$ in the circuit. We will divide the analysis based on the fan-ins of the $f_i$s (\ie, the number of variables they depend on). 

We denote by $\beta$ the quantity ${3}/{4}+ {\varepsilon}/{2}$. It can be checked that we have both
\begin{equation}
\label{eq:beta-props}
\beta = \frac{1}{2}+\varepsilon+\frac{\alpha}{2}+\Omega(\gamma) \qquad\text{and}\qquad  1-\beta = \frac{\alpha}{2}+\Omega(\gamma).
\end{equation}

Consider any $f_i$ of fan-in at most $n^{\beta}$. When hit with a random restriction with $*$-probability $n^{-(1-\alpha)}$, we see that the expected number of variables of $f_i$ that survive is at most $n^{\beta -(1-\alpha)} = n^{\alpha - (1-\beta)} = n^{\alpha/2 - \Omega(\gamma)}$ by (\ref{eq:beta-props}) above. By a Chernoff bound (\expref{Theorem}{thm:chernoff-mul}), the probability that this number exceeds twice its expectation is exponentially small. Union bounding over all the gates of small fan-in, we see that with probability $1-\exp(-n^{\Omega(1)})$, all the low fan-in gates depend on at most $2n^{\alpha/2-\Omega(\gamma)}$ many variables after the restriction. We call this high probability event $\mc{E}_1$.

Now, we consider the gates of fan-in at least $n^{\beta}$. Without loss of generality, let $f_1,\ldots,f_r$ be these LTFs. Since the total number of wires is at most $n^{1+\varepsilon}$, we have $r \leq n^{1+\varepsilon - \beta} = n^{{1}/{2}-{\alpha}/{2} - \Omega(\gamma)}$ by (\ref{eq:beta-props}).

By \expref{Theorem}{thm:Peres}, we know that for any $f_i$,
\[
\avg{\rho}{\mathrm{Var}(f_i|_\rho)} \leq O(\sqrt{p}) = O\left(\frac{1}{n^{(1-\alpha)/2}}\right).
\]
By linearity of expectation, we have
\[
E:=\avg{\rho}{\sum_{i=1}^r\mathrm{Var}(f_i|_\rho)} \leq O\left(r\cdot \frac{1}{n^{(1-\alpha)/2}}\right) = O\left(n^{(1-\alpha)/2-\Omega(\gamma)}\cdot \frac{1}{n^{(1-\alpha)/2}}\right) = O(n^{-\Omega(\gamma)}).
\]

By Markov's inequality, we see that the probability that
\[
  \sum_{i=1}^r\mathrm{Var}(f_i|_\rho) > \sqrt{E}
\]
is at most $\sqrt{E} = n^{-\Omega(\gamma)}$. We let $\mc{E}_2$ denote
the event that
\[
  \sum_{i=1}^r\mathrm{Var}(f_i|_\rho) \leq \sqrt{E}.
\]

Consider the event $\mc{E} = \mc{E}_1 \wedge \mc{E}_2$. A union bound tells us that the probability of $\mc{E}$ is at least $1-n^{-\Omega(\gamma)}$. When this event occurs, we construct the circuit $\tilde{C}_\rho$ from the statement of the claim as follows.

When the event $\mc{E}$ occurs, the LTFs of low arity are already $n^{\alpha/2 - \Omega(\gamma)}$-juntas, so there is nothing to be done for them.

Now, consider the LTFs of high fan-in, which are $f_1,\ldots,f_r$. For each $f_i|_\rho$ ($i\in [r]$), replace $f_i$ by a bit $b_i\in \{-1,1\}$ such that
\[
  \prob{x}{f_i|_\rho(x) = b_i}
  \geq \frac{1}{2}.
\]
In the circuit $\tilde{C}_\rho$, these gates thus become constants, which are $0$-juntas. The circuit $\tilde{C}_\rho$ now has the required form. We now analyze the error introduced by this operation.

We know that
\[
  \prob{x}{f_i|_\rho(x) \neq b_i} \leq 2\mathrm{Var}(f_i|_\rho)
\]
and thus the overall error introduced is at most
\[
2\sum_{i\in [r]}\mathrm{Var}(f_i|_\rho) \leq O(\sqrt{E}) =
n^{-\Omega(\gamma)}
\]
(since $\mc{E}_2$ is assumed to occur). Thus, the circuit
$\tilde{C}_\rho$ is an $n^{-\Omega(\gamma)}$-approximation to $C$.
\end{proof}

\subsection{Correlation bounds for Parity against constant-depth circuits}
\label{subsec:corr-par}

In this section, we prove \expref{Theorem}{thm:corr-parity}, assuming \expref{Lemma}{lem:thresh-gate}. The proof proceeds by iteratively reducing the depth of the circuit. In order to perform this depth-reduction for a depth-$d$ circuit, we need to analyze the threshold map defined by the threshold gates at depth $d-1$. The first observation, which follows from Markov's inequality, shows that we may assume (after setting a few variables) that the map reads each variable only a few times.

\begin{fact}[Small wire-complexity to small number of reads]
\label{fac:reads}
Let $C$ be any threshold circuit on $n$ variables with wire complexity at most $cn$. Then, there is a set $S$ of at most $n/2$ variables such that each variable outside $S$ is an input variable to at most $2c$ many gates in $C$.
\end{fact}

The second observation is that if the fan-ins of all the threshold gates are small, then depth-reduction is easy (after setting some more variables). 

\begin{proposition}[Handling small fan-in gates]
\label{prop:smallfanin}
Let $C = (\phi_1,\ldots,\phi_m)$ be any read-$k$ threshold map on $n$ variables such that $\max_i \fanin(\phi_i) \leq t$. Then, there is a set $S$ of $n/kt$ variables such that each $\phi_i$ depends on at most one variable in $S$.
\end{proposition}

\begin{proof}
This may be done via a simple graph theoretic argument. Define an undirected graph whose vertex set is the set of $n$ variables and two variables are adjacent iff they feed into the same threshold gate. We need to pick an $S$ that is an independent set in this graph. Since the graph has degree at most $kt$, we can greedily find an independent set of size at least $n/kt$. Let $S$ be such an independent set.
\end{proof}

Let $B > 2$ be a constant real parameter that we will choose to satisfy various constraints in the proofs below. For $d\geq 1$, define 
\begin{equation}
\label{eq:epsdeltadefn}
\varepsilon_d = B^{-(2d-1)}\qquad\text{and}\qquad\delta_d = B\varepsilon_d. 
\end{equation}

\begin{reptheorem}{thm:corr-parity}[Correlation bounds for Parity (restated)]
For any $d \geq 1$ and $c\leq n^{\varepsilon_d}$, any  depth-$d$ threshold circuit $C$ with at most $cn$ wires satisfies $\Corr(C,\Par_n)\leq O(n^{-\varepsilon_d})$ where the $O(\cdot)$ hides absolute constants (independent of $d$ and $n$).
\end{reptheorem}

\paragraph{Proof idea.} The proof will be an induction on $d$. Assume that $C$ has $n^{1+\varepsilon}$ many wires for a small $\varepsilon > 0$. Our aim is to apply a restriction that leaves plenty of variables alive (\ie, set to $*$) while at the same time reduces the depth of the circuit from $d$ to $d-1$. 

We first apply a random restriction $\rho$ with $*$-probability $n^{-\delta}$ for a suitable $\delta > 0$ to the circuit $C$ and analyze the effect of $\rho$ on the threshold gates in $C$ at depth $d-1$ (just above the variables). 

By \expref{Lemma}{lem:thresh-gate}, we see that with high probability, each threshold gate $\phi$ at depth $d-1$ becomes imbalanced and hence (by the Chernoff bound (\expref{Theorem}{thm:chernoff})) highly biased. Such gates can be replaced by constants without noticeably changing the correlation with the parity function. Now the only gates at depth $d-1$ are the balanced gates.

It remains to handle the balanced gates. Consider a gate $\phi$ of $C$ with initial fan-in $k$. The expected number of input wires to $\phi$ that survive the random restriction $\rho$ is $k\cdot n^{-\delta}$. Further by \expref{Lemma}{lem:thresh-gate}, the probability that $\phi$ remains biased is at most $n^{-\Omega(\delta)}$. While these two events are not independent, we can argue that their effects are somewhat independent in the following sense. We argue that the random variable $Y$ that is equal to the fan-in of $\phi|_\rho$ if $\phi$ is balanced and $0$ otherwise has expectation roughly $k\cdot n^{-\delta}\cdot n^{-\Omega(\delta)}$. Hence, the expected number of wires $W$ that feed into balanced gates is at most $n^{1+\varepsilon}\cdot n^{-\delta}\cdot n^{-\Omega(\delta)}$ where the first term accounts for the total fan-ins of all the gates at depth $d-1$. For $\delta$ suitably larger than $\varepsilon$, the expectation of $W$ is much smaller than $n^{1-\delta}$, the expected number of surviving variables. Thus, if we set all the variables corresponding to the wires feeding into the balanced gates, we can set all the balanced gates to constants. In the process, we don't set too many live variables (since the expectation of $W$ is much smaller than $n^{1-\delta}$) and further we reduce the depth of the circuit to $d-1$. 

We can now apply the induction hypothesis to bound the correlation of the resulting circuit with the Parity function on (roughly) $n^{1-\delta}$ variables to get the correlation bound against depth-$d$ circuits.

\begin{proof}
The proof is by induction on the depth $d$ of $C$. The base case is $d=1$, which is the case when $C$ is only a single threshold gate. In this case, \expref{Corollary}{cor:Peres} tells us that $\Corr(C,\Par_n)\leq O(n^{-1/2})\leq n^{-\varepsilon_1}$, since $B> 2$.

Now, we handle the inductive case when the depth $d > 1$. Our analysis proceeds in phases.

\paragraph{Phase 1.} We first transform the circuit into a read-$2c$ circuit by setting $n/2$ variables. This may be done by \expref{Fact}{fac:reads}. This defines a restriction tree of depth $n/2$. By \expref{Fact}{fac:corr}, it suffices to show that each leaf of this restriction tree, the correlation of the restricted circuit and $\Par_{n/2}$ remains bounded by $O(n^{-\varepsilon_d})$. 

Let $n_1$ now denote the new number of variables and let $C_1$ now be the restricted circuit at some arbitrary leaf of the restriction tree.  By renaming the variables, we assume that they are indexed by the set $[n_1]$.

\paragraph{Phase 2.} Let $\phi_1,\ldots,\phi_m$ be the threshold gates at depth $d-1$ in the circuit $C_1$. We call $\phi_i$ \emph{large} if $\fanin(\phi_i) > n^{\delta_d}$ and \emph{small} otherwise. Let $L\subseteq [m]$ be defined by $L = \{ i\in [m] \mid \phi_i \text{ large}\}$. Assume that $|L| = \ell$. Note that $\ell\cdot n^{\delta_d}\leq n^{1+\varepsilon_d}$ and hence $\ell\leq n^{1+\varepsilon_d-\delta_d}\leq n$.

We restrict the circuit with a random restriction $\rho = (I,y)\sim \mc{R}_p^{n_1}$, where $p = n^{-\delta_d/2}$. By \expref{Lemma}{lem:thresh-gate}, we know that for each $i\in [m]$ and some $t = {1}/{p^{\Omega(1)}}$ and $q = p^{\Omega(1)}$,
\begin{equation}
\label{eq:corrpar1}
\prob{\rho}{\text{$\phi_i|_\rho$ is $t$-balanced}} \leq q.
\end{equation}

Further, we also know that for each $i\in L$, the expected value of $\fanin(\phi_i|_\rho) = p\cdot \fanin(\phi_i)$, since each variable is set to a constant with probability $1-p$. Since $i\in L$, the expected fan-in of each $\phi_i$ ($i\in L$) is at least $n^{\delta_d/2}$. Hence, by a Chernoff bound (\expref{Theorem}{thm:chernoff-mul}), we see that for any $i\in L$,
\begin{equation}
\label{eq:corrpar2}
\prob{\rho}{\fanin(\phi_i|_\rho) > 2p\cdot \fanin(\phi_i)}\leq \exp(-\Omega(n^{\delta_d/2})).
\end{equation}

Finally another Chernoff bound (\expref{Theorem}{thm:chernoff-mul}) tells us that
\begin{equation}
\label{eq:corrpar3}
\prob{\rho = (I,y)}{|I| < \frac{n_1p}{2}} \leq \exp(-\Omega(n_1p)) = \exp(-\Omega(np)).
\end{equation}

We call a set $I$ \emph{generic} if $|I| \geq {n_1p}/{2}$ and $\fanin(\phi_i|_\rho) \leq 2p\cdot \fanin(\phi_i)$ for each $i\in L$. Let $\mc{G}$ denote the event that $I$ is generic. By  (\ref{eq:corrpar2}) and (\ref{eq:corrpar3}), we know that $$\prob{I}{\neg\mc{G}}\leq \ell\exp(-\Omega(n^{\delta_d/2})) + \exp(-\Omega(np))\leq \exp(-n^{\delta_d/4}).$$ In particular, conditioning on  $\mc{G}$ doesn't change (\ref{eq:corrpar1}) by much.
\begin{equation}
\label{eq:corrpar4}
\prob{\rho = (I,y)}{\text{$\phi_i|_\rho$ is $t$-balanced}\; \middle|\; \mc{G}} \leq q + \exp(-n^{\delta_d/4})\leq 2q.
\end{equation}

Our aim is to further restrict the circuit by setting \emph{all} the input variables to the gates $\phi_i$ that are $t$-balanced. In order to analyze this procedure, we define random variables $Y_i$ ($i\in L$) so that $Y_i=0$ if $\phi_i|_\rho$ is $t$-imbalanced and $\fanin(\phi_i|_\rho)$ otherwise. Let $Y = \sum_{i\in L}Y_i$. Note that 
\[
  \avg{\rho}{Y_i\ |\ \mc{G}}\leq (2p\cdot \fanin(\phi_i))\cdot \prob{\rho}{\text{$\phi_i|_\rho$ is $t$-balanced}\;\middle|\; \mc{G}}\leq 4pq\cdot \fanin(\phi_i)
\]
where the first inequality follows from the fact that since we have conditioned on $I$ being generic, we have
\[
  \fanin(\phi_i|_\rho)\leq 2p\cdot \fanin(\phi_i)
\]
with probability $1$. Hence, we have
\begin{equation}
\label{eq:corrpar5}
\avg{\rho}{Y\;\middle|\; \mc{G}} \leq 4pq\cdot \sum_i\fanin(\phi_i)\leq 4pq\cdot n^{1+\varepsilon_d}.
\end{equation}
We let $\mu := 4pq\cdot n^{1+\varepsilon_d}$. By Markov's inequality,
\begin{equation}
\label{eq:corrpar6}
\prob{\rho}{Y \geq \frac{\mu}{\sqrt{q}}\;\middle|\; \mc{G}} \leq \sqrt{q}.
\end{equation}

In particular, we can condition on a \emph{fixed} generic $I\subseteq [n]$ such that for random $y\sim \{-1,1\}^{n_1-|I|}$, we have
\[
\prob{y}{Y \geq \frac{\mu}{\sqrt{q}}}\leq \sqrt{q}.
\]

The above gives us a restriction tree $T$ (that simply sets all the variables in $[n_1]\setminus I$) such that at all but $1-2\sqrt{q}$ fraction of leaves $\lambda$ of $T$, the total fan-in of the large gates at depth $1$ in $C_1$ that are $t$-balanced is at most ${\mu}/{\sqrt{q}}$; call such $\lambda$ \emph{good} leaves. Let $n_2$ denote $|I|$, which is the number of surviving variables. 

\paragraph{Phase 3.} We will show that for any good leaf $\lambda$, we have 
\begin{equation}
\label{eq:corrpar7}
\Corr(C_\lambda,\Par_{n_2})\leq n^{-\varepsilon_d}
\end{equation}
where $C_\lambda$ denotes $C_1|_{\rho_\lambda}$. This will prove the theorem, since we have by \expref{Fact}{fac:corr},
\begin{align*}
\Corr(C_1,\Par_{n_1}) &\leq \avg{\lambda\sim T}{\Corr(C_\lambda,\Par_{n_2})} \\
&\leq \prob{\lambda}{\text{$\lambda$ not good}} + \max_{\text{$\lambda$ good}} \Corr(C_\lambda,\Par_{n_2})\\
&\leq 2\sqrt{q}+ n^{-\varepsilon_d} \leq 2n^{-\varepsilon_d}
\end{align*}
where we have used the fact that $\Par_{n_1}|_{\rho_\lambda} = \pm\Par_{n_2}$ for each leaf $\lambda$, and also that $2\sqrt{q} \leq n^{-\varepsilon_d}$ for a large enough choice of the constant $B$.

It remains to prove (\ref{eq:corrpar7}). We do this in two steps.

In the first step, we set all large $t$-imbalanced gates to their most probable constant values. Formally, for a $t$-imbalanced threshold gate $\phi$ labeled by $(w,\theta)$, we have $|\theta| \geq t\cdot \mynorm{w}_2$. We replace $\phi$ by a constant $b_\phi$ which is $1$ if $\theta \geq t\cdot \mynorm{w}_2$ and by $-1$ if $-\theta\geq t\cdot \mynorm{w}_2$. This turns the circuit $C_\lambda$ into a circuit $C'_\lambda$ of at most the wire complexity of $C_\lambda$. Further, note that for any $x\in \{-1,1\}^{n_1}$, $C_\lambda(x) = C'_\lambda(x)$ unless there is a $t$-imbalanced threshold gate $\phi$ such that $\phi(x)\neq b_{\phi}(x)$. By the Chernoff bound (\expref{Theorem}{thm:chernoff}) the probability that this happens for any fixed imbalanced threshold gate is at most
\[
  \exp(-\Omega(t^2))\leq \exp(-n^{\Omega(\delta_d)}).
\]
By a union bound over the $\ell\leq n$ large threshold gates, we see
that
\[
  \prob{x}{C_\lambda(x)\neq C'_\lambda(x)}\leq
  n\exp(-n^{\Omega(\delta_d)}).
\]
In particular, we get by
\expref{Fact}{fac:corr}
\begin{equation}
\label{eq:corrpar8}
\Corr(C_\lambda,\Par_{n_2})\leq \Corr(C_{\lambda}',\Par_{n_2}) + n\exp(-n^{\Omega(\delta_d)})\leq \Corr(C_{\lambda}',\Par_{n_2}) + \exp(-n^{\varepsilon_d})
\end{equation}
where the last inequality is true for a large enough constant $B$.

In the second step, we further define a restriction tree $T_\lambda$ such that $C'_\lambda$ becomes a depth-$(d-1)$  circuit with at most $cn$ wires at \emph{all the leaves} of $T_\lambda$. We first restrict by setting all variables that feed into \emph{any} of the $t$-balanced gates. The number of variables set in this way is at most 
\[
\frac{\mu}{\sqrt{q}}\leq 4p\sqrt{q}\cdot n^{1+\varepsilon_d} \leq (pn)\cdot (4\sqrt{q}n^{\varepsilon_d}) \leq \frac{pn}{8}\leq \frac{n_2}{2}
\]
for a large enough choice of the constant $B$. This leaves $n_3\geq {n_2}/{2}$ variables still alive. Further, all the large $t$-balanced gates are set to constants \emph{with probability $1$}. Finally, by \expref{Proposition}{prop:smallfanin}, we may set all but a set $S$ of $n_4 = n_3/2cn^{\delta_d}$ variables to ensure that with probability $1$, all the small gates depend on at most one input variable each. At this point, the circuit $C'_\lambda$ may be transformed to a depth-$(d-1)$ circuit $C''_\lambda$ with at most as many wires as $C'_\lambda$, which is at most $cn$. 

Note that the number of unset variables is $n_4 \geq pn/8cn^{\delta_d}\geq n^{1-2\delta_d}$, for large enough $B$. Hence, the number of wires is at most $$cn \leq n_4^{\frac{1+\varepsilon_d}{1-2\delta_d}}\leq n_4^{(1+\varepsilon_d)(1+3\delta_d)}\leq n_4^{1+\varepsilon_{d-1}}$$ for suitably large $B$. Thus, by the inductive hypothesis, we have
\[
\Corr(C''_\lambda,\Par_{n_4})\leq O(n_4^{-\varepsilon_{d-1}})\leq n^{-\varepsilon_d}/2
\]
with probability $1$ over the choice of the variables restricted in the second step. Along with (\ref{eq:corrpar8}) and \expref{Fact}{fac:corr}, this implies (\ref{eq:corrpar7}) and hence the theorem.
\end{proof}

\subsection{Exponential correlation bounds for the Generalized Andreev function}
\label{subsec:corr-andreev}

We now prove an exponentially strong correlation bound for the Generalized Andreev function defined in \expref{Section}{sec:andreev} with any $\gamma < 1/6$. For any $d\geq 1,$ we define the constants $\varepsilon_d$ and $\delta_d$ as in~\eqref{eq:epsdeltadefn}, where $B > 2$ is a large constant that will be chosen below.

\paragraph{Proof overview.} As in the case of \expref{Theorem}{thm:corr-parity}, the proof proceeds by an iterative depth reduction. The proof of the depth reduction here is technically more challenging than the proof of \expref{Theorem}{thm:corr-parity} since we are trying to prove exponentially small correlation bounds and hence, we can no longer afford to ignore ``bad'' events that occur with polynomially small probability. This forces us to prove a more involved depth-reduction statement, which we prove as a separate lemma (\expref{Lemma}{lem:depthredn} below). We show how to use this depth-reduction statement to prove our correlation bound in \expref{Theorem}{thm:genstrong} and \expref{Corollary}{cor:corr-andreev}. 

We begin with a definition that will be required to state our depth-reduction lemma.

\begin{definition} [Simplicity]
We call a threshold circuit $C$ \emph{$(t,d,w)$-simple} if there is a set $R$ of $r\leq t$ threshold functions $g_1,\ldots,g_r$ such that for every setting of these threshold functions to bits $b_1,\ldots,b_r$, the circuit $C$ can be represented on the corresponding inputs (\ie, inputs $x$ satisfying $g_i(x) = b_i$ for each $i\in [r]$) by a depth-$d$ threshold gate of wire complexity at most $w$. 
\end{definition}

In particular, note that a $(t,d,w)$-simple circuit $C$ may be expressed as
\begin{equation}
\label{eq:simple}
C(x) = \bigvee_{b_1,\ldots,b_r\in \{-1,1\}} \left(C_{b_1,\ldots,b_r}\wedge \bigwedge_{i: b_i = -1} g_i \wedge \bigwedge_{i: b_i = 1} (\neg g_i)\right).
\end{equation}
where each $C_{b_1,\ldots,b_r}$ is a depth-$d$ circuit of wire complexity at most $w$. Further, note that the OR appearing in the above expression is disjoint (\ie, no two terms in the OR can be simultaneously true).

We need the following elementary property of disjoint ORs.

\begin{proposition}
\label{prop:disjOR}
Let $g_1,\ldots,g_N$ be Boolean functions defined on $\{-1,1\}^n$ such that no two of them are simultaneously true and let $h$ denote $\bigvee_{i\in [N]}g_i$. Then, for any $f:\{-1,1\}^n\rightarrow \{-1,1\}$, we have
\[
  \Corr(f,h)\leq (N-1)\Corr(f,1)+\sum_{i\in [N]}\Corr(f,g_i),
\]
where $1$ denotes the constant $1$ function.
\end{proposition}

\begin{proof}
Note that $h$ being the disjoint OR of $g_1,\ldots,g_N$ translates into the following identity:
\[
\frac{1-h}{2} = \sum_{i\in [N]}\frac{1-g_i}{2}.
\]
Hence, we have $h = \left(\sum_i g_i \right) -(N-1)$. Thus, using the bilinearity of the inner product and the triangle inequality
\[
\Corr(f,h) = |\ip{f}{h}| \leq \sum_{i\in [N]} \Corr(f,g_i) + (N-1) \Corr(f,1).\qedhere
\]
\end{proof}

We are now ready to prove our main depth-reduction lemma for threshold circuits with small wire complexity.

\begin{lemma}[Depth-reduction lemma with exponentially small failure probability]
\label{lem:depthredn}
Let $d\geq 1$ be any constant and assume that $\varepsilon_d,\delta_d$ are defined as in (\ref{eq:epsdeltadefn}) with $B$ large enough. Say we are given any depth-$d$ threshold circuit $C$ on $n$ variables with at most $n^{1+\varepsilon_d}$ wires. 

There is a restriction tree $T$ of depth $n - n^{1-2\delta_d}$ with the following property: for a random leaf $\lambda \sim T$, let $\mc{E}(\lambda)$ denote the event that the circuit $C|_{\rho_\lambda}$ is $\exp(-n^{\varepsilon_d})$-approximated by an $(n^{\delta_d},d-1,n^{1+\varepsilon_d})$-simple circuit. Then, $\prob{\lambda}{\neg \mc{E}(\lambda)}\leq \exp(-n^{\varepsilon_d})$.
\end{lemma}

\paragraph{Comparison with \expref{Theorem}{thm:corr-parity}.} The proof of this lemma is similar to the proof of the depth-reduction in \expref{Theorem}{thm:corr-parity}, with a few key differences. The main change is that we require the simplification of depth-$d$ to depth-$(d-1)$ to fail only with \emph{exponentially} small probability as opposed to the bound of $n^{-\Omega_d(1)}$ in the proof of \expref{Theorem}{thm:corr-parity}. To ensure this, we will use the read-$k$ Chernoff bound stated in \expref{Theorem}{thm:GLSS}. 

However, this lemma yields strong bounds only when the number of random variables (the parameter $m$ in the statement of \expref{Theorem}{thm:GLSS}) is quite large. In particular, for example, if the number of gates at depth $d-1$ is small (a constant, say), then the bound obtained from \expref{Theorem}{thm:GLSS} is not very useful. In this case, instead of simplifying these gates, we simply say that for each setting of the outputs of these gates, we get a circuit of smaller depth. This is why we obtain $(t,d-1,w)$-simple circuits in general and not simply a depth-$(d-1)$ threshold circuit.

\begin{proof}
Let $\phi_1,\ldots,\phi_m$ be the threshold gates appearing at depth $d-1$ in the circuit $C$. We say that $\phi_i$ is large if $\fanin(\phi_i)\geq n^{\delta_d}$ and small otherwise. Let $L = \{i  \mid \text{$\phi_i$ large}\}$ and $S = [m]\setminus L$. Let $\ell = |L|$. Note that $\ell \leq n^{1+\varepsilon_d-\delta_d}\leq n$. Let $c = n^{\varepsilon_d}$.

As in the inductive case of \expref{Theorem}{thm:corr-parity}, our construction proceeds in phases. 

\paragraph{Phase 1.} This is identical to Phase 1 in \expref{Theorem}{thm:corr-parity}. We thus get a restriction tree of depth $n/2$ such that at \emph{all} leaves of this tree, the resulting circuit is a read-$2c$ circuit with at most $cn$ wires. Let $C_1$ denote the circuit obtained at some arbitrary leaf of the restriction tree and let $n_1$ denote the number of variables.

\paragraph{Phase 2.} This basic idea here is similar to Phase $2$ from \expref{Theorem}{thm:corr-parity}. However, there are technical differences from \expref{Theorem}{thm:corr-parity} since we apply a concentration bound to ensure that the circuit simplifies with very high probability. 

We restrict the circuit with a random restriction $\rho = (I,y)\sim \mc{R}_p^{n_1}$, where $p = n^{-\delta_d/2}$. As in \expref{Theorem}{thm:corr-parity}, we have for some $t = {1}/{p^{\Omega(1)}}$, $q = p^{\Omega(1)}$, and for each $i\in L$,
\begin{align}
\prob{\rho}{\text{$\phi_i|_\rho$ is $t$-balanced}} &\leq q, \label{eq:depthredn1}\\
\prob{\rho}{\fanin(\phi_i|_\rho) > 2p\cdot \fanin(\phi_i)} &\leq \exp(-\Omega(n^{\delta_d/2})),\qquad\text{and}\label{eq:depthredn2}\\
\prob{\rho = (I,y)}{|I| < \frac{n_1p}{2}} &\leq \exp(-\Omega(np)).\label{eq:depthredn3}
\end{align}

Now, we partition $L$ as $L = L_1\cup \cdots \cup L_a$, where $a \leq {1}/{\varepsilon_d}$, as follows. The set $L_j$ indexes all threshold gates at depth $d-1$ of fan-in at least $n^{\delta_d + (j-1)\varepsilon_d}$ and less than $n^{\delta_d + j\varepsilon_d}$. We let $\ell_j$ denote $|L_j|$. For each $i\in L$, let $Y_i$ be a random variable that is $1$ if $\phi_i|_\rho$ is $t$-balanced and $0$ otherwise. Note that this defines a collection of read-$2c$ Boolean random variables (the underlying independent random variables are $\rho(k)$ for each $k\in [n_1]$).

Let $Z_j = \sum_{i\in L_j} Y_i$, the number of $t$-balanced gates in $L_j$. We have
\[
  \avg{}{Z_j} = \sum_{i\in L_j}\avg{}{Y_i}\leq q\ell_j
\]
by (\ref{eq:depthredn1}). Thus, by an application of the read-$2c$
Chernoff bound in \expref{Theorem}{thm:GLSS}, we have
\[
\prob{}{Z_j \geq 2q\ell_j} \leq \exp\{-\Omega(q\ell_j/c)\}.
\]
Assuming that $\ell_j\geq n^{3\delta_d/4}$ and $B = \delta_d/\varepsilon_d$ is a large enough constant, the right hand side of the above inequality is upper bounded bounded by $\exp\{-2n^{\varepsilon_d}\}$. On the other hand if $\ell_j < n^{3\delta_d/4}$, then $Z_j < n^{3\delta_d/4}$ with probability $1$. Hence, we have
\[
  \prob{\rho = (I,y)}{Z_j \geq \max\{2q\ell_j,n^{3\delta_d/4}\}} \leq \exp\{-2n^{\varepsilon_d}\}
\]
and by a union bound
\begin{equation}
\label{eq:depthredn4}
\prob{\rho = (I,y)}{\exists j\in [a], Z_j \geq \max\{2q\ell_j,n^{3\delta_d/4}\}} \leq a\exp\{-2n^{\varepsilon_d}\}.
\end{equation}

We call a set $I$ \emph{generic} if $|I| \geq {n_1p}/{2}$ and $\fanin(\phi_i|_\rho) \leq 2p\cdot \fanin(\phi_i)$ for each $i\in L$. Let $\mc{G}$ denote the event that $I$ is generic. By  (\ref{eq:depthredn2}) and (\ref{eq:depthredn3}), we know that $$\prob{I}{\neg\mc{G}}\leq \ell\exp(-\Omega(n^{\delta_d/2})) + \exp(-\Omega(np))\leq \exp(-n^{\delta_d/4}).$$ In particular, similar to \expref{Theorem}{thm:corr-parity}, we get,
\begin{equation}
\label{eq:depthredn5}
\prob{\rho = (I,y)}{\exists j\in [a], Z_j \geq \max\{2q\ell_j,n^{3\delta_d/4}\}\; \middle|\; \mc{G}} \leq a\exp\{-2n^{\varepsilon_d}\} + \exp(-n^{\delta_d/4}) \leq 2a\exp\{-2n^{\varepsilon_d}\}.
\end{equation}

We fix any generic $I$ such that
\begin{equation}
\label{eq:depthredn5.5}
\prob{y}{\exists j\in [a], Z_j \geq \max\{2q\ell_j,n^{3\delta_d/4}\}}\leq 2a\exp\{-2n^{\varepsilon_d}\}.
\end{equation}

Consider the restriction tree $T$ that sets all the variables not in $I$. The tree leaves $n_2 \geq pn_1/2 = pn/4$ variables unfixed. We call a leaf $\lambda$ of the tree \emph{good} if for each $j\in [a]$ we have $Z_j < \max\{2q\ell_j,n^{3\delta_d/4}\}$ and \emph{bad} otherwise. We have
\begin{equation}
\label{eq:depthredn6}
\prob{\lambda\sim T}{\text{$\lambda$ a bad leaf}}\leq 2a \exp(-2n^{\varepsilon_d}).
\end{equation}

For good leaves $\lambda$, we show how to approximate $C_{\lambda}:= C_1|_{\rho_\lambda}$ as claimed in the lemma statement.

For the remainder of the argument, fix any good leaf $\lambda$. We partition $[a] = J_1\cup J_2$ where
\[
  J_1 = \{j\in [a]\ |\ Z_j < 2q\ell_j\}\,.
\]
Note that for any $j\in J_1$, we have
\begin{align*}
\sum_{i\in L_j}Y_i\cdot \fanin(\phi_i|_{\rho_\lambda}) &\leq \sum_{i\in L_j}Y_i\cdot 2p\cdot \fanin(\phi_i)\\
&\leq 2p\cdot n^{\delta_d + j\cdot \varepsilon_d}\cdot Z_j
\leq n^{\delta_d + j\cdot \varepsilon_d}\cdot 4pq\ell_j\\
&= 4pq n^{\varepsilon_d}\cdot \ell_j\cdot n^{\delta_d + (j-1)\cdot \varepsilon_d}\leq 4pqn^{\varepsilon_d}n^{1+\varepsilon_d}\\
&= 4pqn^{1+2\varepsilon_d} 
\end{align*}
where for the last inequality, we have used the fact that since we have $\ell_j$ gates of fan-in at least $n^{\delta_d + (j-1)\varepsilon_d}$ each, we must have $\ell_j\cdot n^{\delta_d + (j-1)\varepsilon_d}\leq n^{1+\varepsilon_d}$, the total wire complexity of the circuit.

In particular, we can bound the total fan-in of all the $t$-balanced gates indexed by $\bigcup_{j\in J_1}L_j$ by
\begin{equation}
\label{eq:depthredn7}
\sum_{j\in J_1}\sum_{i\in L_j}Y_i\cdot \fanin(\phi_i|_{\rho_\lambda}) \leq \frac{4pqn^{1+2\varepsilon_d}}{\varepsilon_d}.
\end{equation}
\paragraph{Phase 3.} We proceed in two steps as in \expref{Theorem}{thm:corr-parity}. Since the steps are very similar, we just sketch the arguments. In the first step, we replace all large $t$-imbalanced gates by their most probable values. This yields a circuit $C'_\lambda$ of at most the wire complexity of $C_\lambda$ and such that 

\begin{equation}
\label{eq:depthredn8}
\prob{x}{C_\lambda(x)\neq C'_\lambda(x)} \leq \ell\exp(-n^{\Omega(\delta_d)})\leq n\exp(-n^{\Omega(\delta_d)})\leq \exp(-n^{\varepsilon_d}).
\end{equation}

In the second step, we construct another restriction tree rooted at $\lambda$ that simplifies the circuit to the required form. We first restrict by setting all variables that feed into  the $t$-balanced gates that are indexed by $\bigcup_{j\in J_1}L_j$. By (\ref{eq:depthredn7}), the number of variables set is bounded by 
\[
\frac{4pqn^{1+2\varepsilon_d}}{\varepsilon_d} \leq \frac{4pn\cdot n^{\varepsilon_d -\Omega(\delta_d)}}{\varepsilon_d} \leq \frac{pn}{8}\leq \frac{n_2}{2}
\]
for a large enough choice of the constant $B$. This sets all the $t$-balanced gates indexed by $\bigcup_{j\in J_1}L_j$ to constants while leaving $n_3\geq {n_2}/{2}$ variables still alive.  Finally, by \expref{Proposition}{prop:smallfanin}, we may set all but a set of $n_4 = n_3/2cn^{\delta_d}$ variables to ensure that with probability $1$, all the small gates depend on at most one input variable each. We may replace the small gates by the unique variable they depend on or a constant (if they do not depend on any variable) without increasing the wire complexity of the circuit. Call the circuit thus obtained $C''_\lambda$.

At this point, the only threshold gates at depth $d-1$ in the circuit $C''_\lambda$ are the gates indexed by the $t$-balanced gates in $\bigcup_{j\in J_2}L_j$. But by the definition of $J_2$, there can be at most $({1}/{\varepsilon_d})\cdot n^{3\delta_d/4}\leq n^{\delta_d}$ of them. For every setting of these threshold gates to constants, the circuit becomes a depth-$(d-1)$ circuit of size at most $n^{1+\varepsilon_d}$. Hence, we have a $(n^{\delta_d},d-1,n^{1+\varepsilon_d})$-simple circuit, as claimed.

Note that the number of variables still surviving is given by $n_4 \geq pn/16cn^{\delta_d}\geq n^{1-2\delta_d}$, for a large enough choice of the parameter $B$. Hence, the restriction tree constructed satisfies the required depth constraints. 

For a random leaf $\nu\sim T$, the probability $\mc{E}(\nu)$ does not occur is at most the probability that in Phase 2, the leaf sampled is bad. By (\ref{eq:depthredn6}), this is bounded by $2a\exp(-2n^{\varepsilon_d})\leq \exp(-n^{\varepsilon_d})$ as claimed.
\end{proof}

We now prove the correlation bound for the Generalized Andreev function against threshold circuits with small wire complexity. For the sake of induction, it helps to prove a statement that is stronger in two ways: firstly, we  consider any function $F_a = F(a,\cdot)$ where $a\in \{-1,1\}^{4n}$ has high Kolmogorov complexity and the input to $F_a$ is further restricted by an arbitrary restriction $\rho$ that leaves a certain number of variables alive; secondly, we prove a correlation bound against circuits which are the AND of a small threshold circuit with a small number of threshold gates. 

\begin{definition}[Intractability]
We say that $f:\{-1,1\}^n\rightarrow \{-1,1\}$ is \emph{$(N,d,t,\alpha)$-intractable} if for any restriction $\rho$ on $n$ variables that leaves $m\geq N$ variables unset, any depth-$d$ threshold circuit $C$ on $m$ variables of wire complexity at most $m^{1+\varepsilon_d}$, and any set $S$ of at most $t$ threshold functions, we have
\[
\Corr\Bigl(f,C\wedge \bigwedge_{g\in S} g\Bigr)\leq \alpha.
\]
\end{definition}

The main theorem of this section is the following generalized correlation bound.

\begin{theorem}[Generalized version of strong correlation]
\label{thm:genstrong}
Fix any constant $d\geq 1$. Let $F$ be the Generalized Andreev function on $5n$ variables as defined in \expref{Section}{sec:andreev} with $\gamma$ being any constant smaller than $1/6$. Let $a\in \{-1,1\}^{4n}$ be any string with $K(a) \geq 3n$. Then, the function $F_a$ is $(n^{1-\varepsilon_d},d,n^{\varepsilon_d},\exp(-n^{\varepsilon_d}/2))$-intractable.
\end{theorem}

The proof is by induction on $d$. The properties of $F_a$ are only used to prove the base case of the theorem, which can then be used to prove the induction case using \expref{Lemma}{lem:depthredn}. We prove the base case separately below (we assume that the constant $B>0$ is large enough so that this implies the base case of the theorem stated above).

\begin{lemma}[Base case of induction]
\label{lem:base-case}
Let $a\in \{-1,1\}^{4n}$ be any string with $K(a) \geq 3n$. Then, the function $F_a$ is $(\sqrt{n},1,\sqrt{n},\exp(-n^{\Omega(1)}))$-intractable.
\end{lemma}

\begin{proof}
Let $\gamma < 1/6$ in the definition of the Generalized Andreev function in \expref{Section}{sec:andreev}.
Let $\tau$ be any restriction of $n$ variables leaving $m\geq \sqrt{n}$ variables unfixed.
Define $f  := F_a|_\tau$.
Let $C$ be a conjunction of  $\sqrt{n} + 1$ threshold gates each on $m$ variables.
We wish to prove that
\[
	\Corr(f, C)  \leq \exp(-n^{\Omega(\gamma)}).
\]

We build a restriction tree $T$ for $C$ of depth $m -n^\gamma$, by restricting all but $n^\gamma$ arbitrarily chosen variables.
For any leaf $\ell$ of $T$, the restricted circuit $C_\ell := C|_{\rho_\ell}$ is a conjunction of $\sqrt{n} + 1$ threshold gates each on $n^\gamma$ variables. By \expref{Corollary}{cor:numthr}, each threshold function can be described using $n^{2\gamma}$ bits. Hence, the entire circuit can be described in a standard way using $(\sqrt{n} + 1)\cdot O(n^{2\gamma})< n$ bits. Then, by \expref{Lemma}{lem:ckksz-corr}, we have 
\[
  \Corr\bigl(f|_{\rho_\ell},C_\ell\bigr)\leq  \exp( - {n^{\Omega(\gamma)}}).
\]
By \expref{Fact}{fac:corr}, we then obtain $\Corr(f,C)\leq \exp(n^{-\Omega(\gamma)})$.
\end{proof}

\begin{proof}[Proof of \expref{Theorem}{thm:genstrong}]
  We only need to prove the inductive case. Assume that $d\geq 2$ is given. Fix any restriction $\rho$ that sets all but $m\geq n^{1-\varepsilon_d}$ variables and let $f = F_a|_\rho$. Let $C$ be a depth-$d$ threshold circuit on the surviving variables of wire complexity at most $m^{1+\varepsilon_d}$. Let $S$ be any set of at most $n^{\varepsilon_d}$ threshold functions on the $m$ variables. We need to show that
  \[
    \Corr\Bigl(f,C\wedge \bigwedge_{g\in S} g\Bigr)\leq \exp(-n^{\varepsilon_d/2}).
  \]
  
Apply \expref{Lemma}{lem:depthredn} to circuit $C$ to find a restriction tree $T$ as guaranteed by the statement of the lemma. By \expref{Fact}{fac:corr}, we have
\begin{align}
\Corr\Bigl(f,C\wedge \bigwedge_{g\in S} g\Bigr) &\leq \avg{\ell\sim T}{\Corr\Bigl(f_\ell, C_\ell \wedge \bigwedge_{g\in S} g_\ell\Bigr)}\notag\\
&\leq \prob{\ell}{\neg\mc{E}(\ell)} + \max_{\text{$\ell$:$\mc{E}(\ell)$ holds}} \Corr\Bigl(f_\ell, C_\ell \wedge \bigwedge_{g\in S} g_\ell\Bigr)\label{eq:genstrong1}
\end{align}
where $f_\ell$ denotes $f|_{\rho_\ell}$ and similarly for $C_\ell$ and $g_\ell$, and $\mc{E}(\ell)$ is the event defined in the statement of \expref{Lemma}{lem:depthredn}. 

Fix any leaf $\ell$ so that $\mc{E}(\ell)$ holds. We want to bound
\[
  \Corr\Bigl(f_\ell, C_\ell \wedge \bigwedge_{g\in S} g_\ell\Bigr).
\]
By definition of $\mc{E}(\ell)$, we know that $C_\ell$ is
$\exp(-m^{\varepsilon_d})$-approximated by a
$(m^{\delta_d},d-1,m^{1+\varepsilon_d})$-simple circuit
$C'_\ell$. This implies that $C_\ell \wedge \bigwedge_{g\in S} g_\ell$
is $\exp(-m^{\varepsilon_d})$-approximated by
$C'_\ell\wedge \bigwedge_{g\in S} g_\ell$. Hence, we have

\begin{equation}
\label{eq:genstrong2}
\Corr\Bigl(f_\ell, C_\ell \wedge \bigwedge_{g\in S} g_\ell\Bigr) \leq \Corr\Bigl(f_\ell, C_\ell' \wedge \bigwedge_{g\in S} g_\ell\Bigr) + \exp(-m^{\varepsilon_d}).
\end{equation}

Further, by the definition of simplicity and its consequence (\ref{eq:simple}), we know that there exist $r\leq m^{\delta_d}$ threshold functions $h_1^{\ell},\ldots, h_r^{\ell}$ such that 
\[
C'_\ell = \bigvee_{b_1,\ldots,b_r\in \{-1,1\}} C_{b_1,\ldots,b_r}\wedge \bigwedge_{i:b_i = -1} h_i^\ell \wedge \bigwedge_{i: b_i = 1}\neg h_i^\ell
\]
where each $C_{b_1,\ldots,b_r}$ is a depth-$(d-1)$ threshold circuit of size at most $m^{1+\varepsilon_d}$ and the OR above is disjoint. This further implies that
\begin{equation}
\label{eq:genstrong3}
C'_\ell\wedge \bigwedge_{g\in S} g_\ell = \bigvee_{b_1,\ldots,b_r\in \{-1,1\}} \left(C_{b_1,\ldots,b_r}\wedge \bigwedge_{i:b_i = -1} h_i^\ell \wedge \bigwedge_{i: b_i = 1}\neg h_i^\ell \wedge \bigwedge_{g\in S} g_\ell\right)
\end{equation}
and the OR remains disjoint. 

Note that we may apply the induction hypothesis to obtain a bound on the correlation with each term in the OR at this point, since the number of surviving variables is at least $m_1 = m^{1-2\delta_d}\geq n^{1-\varepsilon_d - 2\delta_d}\geq n^{1-\varepsilon_{d-1}}$ (throughout, we assume that $B$ is a large enough constant for many of the inequalities to hold); and the wire complexity of each depth-$(d-1)$ circuit $C_{b_1,\ldots, b_r}$ is at most $$m^{1+\varepsilon_d}\leq m_1^{(1+\varepsilon_d)/(1-2\delta_d)}\leq m_1^{1+\varepsilon_d + 3\delta_d}\leq m_1^{1+\varepsilon_{d-1}}$$ and further, the number of threshold functions in each term is at most $n^{\varepsilon_d}+n^{\delta_d} < m^{\varepsilon_{d-1}}$. Thus, by the inductive hypothesis, we obtain for any $b_1,\ldots,b_r$,
\[
\Corr\Bigl(f,C_{b_1,\ldots,b_r}\wedge \bigwedge_{i:b_i = -1} h_i^\ell \wedge \bigwedge_{i: b_i = 1}\neg h_i^\ell \wedge \bigwedge_{g\in S} g_\ell\Bigr) \leq \exp(-n^{\varepsilon_{d-1}/2}).
\]

Using the fact that the OR in (\ref{eq:genstrong3}) is disjoint, from \expref{Proposition}{prop:disjOR}, we obtain
\[
\Corr\Bigl(f,C'_\ell\wedge \bigwedge_{g\in S} g_\ell\Bigr) \leq O(2^r\cdot \exp(-n^{\varepsilon_{d-1}/2})) \leq O\bigl(2^{n^{\delta_d}}\cdot \exp(-n^{\varepsilon_{d-1}/2})\bigr)\leq \exp(-n^{\varepsilon_d}).
\]
Putting the above together with (\ref{eq:genstrong1}) and (\ref{eq:genstrong2}), we obtain
\[
\Corr\Bigl(f,C\wedge \bigwedge_{g\in S}g\Bigr)\leq \exp(-m^{\varepsilon_d}) + \exp\bigl(-n^{\varepsilon_d}\bigr) \leq \exp(-n^{\varepsilon_d/2}).
\]
which proves the induction case and hence the theorem.
\end{proof}

\begin{repcor}{cor:corr-andreev}[Correlation bounds for Andreev's function (restated)]
Fix any $d \geq 1$ and let $F$ be as in the statement of \expref{Theorem}{thm:genstrong}. Any  depth-$d$ threshold circuit $C$ of wire complexity at most $n^{1+\varepsilon_d}$ satisfies $\Corr(C,F)\leq 2\exp(-n^{\varepsilon_d/2})$.
\end{repcor}

\begin{proof}
For a random $a\in \{-1,1\}^{4n}$, we know by \expref{Fact}{fac:kol-lbd} that $K(a)\geq 3n$ with probability $1-\exp(-\Omega(n))$. For each such $a$, by \expref{Theorem}{thm:genstrong}, we have $\Corr(C_a,F_a)\leq \exp(-n^{\varepsilon_d/2})$, where $C_a$ is the circuit obtained by substituting $x_1 = a$ in $C$. Hence, we have
\[
\Corr(C,F) \leq \avg{a}{\Corr(C_a,F_a)}\leq \exp(-\Omega(n)) + \exp(-n^{\varepsilon_d/2}) \leq 2\exp(-n^{\varepsilon_d/2})
\]
as claimed.
\end{proof}

\section{Proof of Main Structural Lemma (\expref{Lemma}{lem:thresh-gate})}
\label{subsec:msl}

We prove the following statement that was key in proving the results of \expref{Section}{sec:wires}. 

\begin{replemma}{lem:thresh-gate}[Main Structural lemma for threshold gates (restated)]
The following holds for some absolute constant $p_0\in [0,1]$. For any threshold gate $\phi$ over $n$ variables with label $(w,\theta)$ and any $p\in [0,p_0]$, we have\footnote{Recall from \expref{Definition}{def:restrict} that $\phi|_\rho$ is the threshold gate obtained by setting variables outside $\rho^{-1}(*)$ according to $\rho$.}
\[
\prob{\rho\sim \mc{R}^n_p}{\text{$\phi|_\rho$ is $\frac{1}{p^{\Omega(1)}}$-balanced}} \leq p^{\Omega(1)}.
\]
\end{replemma}

\paragraph{Proof overview.} The proof follows a template that has been used in many results on threshold functions (see, \eg,~\cite{Ser,OS,DGJSV,HKM,DRST,MZ}). The basic idea is to divide threshold gates into one of two kinds: threshold gates such that none of their variables have too much weight (called \emph{regular} threshold functions below) and those where some variable has large weight. For regular gates, the lemma can be easily proved by appealing to standard results on anticoncentration of sums of independent random variables. For gates that are not regular, we first try to make them regular by setting a few variables. If we succeed, then we can appeal to the regular case and we will be done. 
Otherwise, it must be the case that there are many variables that have large weight in the threshold gate. In this case, we show that the random restriction, with high probability, sets a large number of these variables. Since these variables have large weight, this results in a potentially large shift in the threshold $\theta$ of the threshold gate, causing it to become imbalanced with high probability.

We now proceed with the formal details. We need the following definitions and facts that have appeared many times before in the literature on threshold functions (see, \eg,~\cite{DGJSV}).

\begin{definition}[Regularity and Critical index]
Let $\varepsilon \in [0,1]$ be a real parameter. We say that $w\in \mathbb{R}^n$ is $\varepsilon$-regular if for each $i\in [n]$, $|w_i|\leq \varepsilon\cdot \mynorm{w}_2$. 

Assume that the co-ordinates of the vector $w$ are sorted so that $|w_1|\geq |w_2|\geq \cdots \geq |w_n|$. Let $w_{\geq i}\in \mathbb{R}^{n-i+1}$ denote the vector obtained by removing the first $i-1$ co-ordinates of $w$. We define the $\varepsilon$-critical index of $w$ be the least $K = K(\varepsilon)$ so that the vector $w_{\geq K+1}$ is $\varepsilon$-regular. Note that $K = 0$ if $w$ is already $\varepsilon$-regular and we define $K=n$ if the $\varepsilon$-critical index is not defined.

We say that an $n$-variable threshold gate $\phi$ labeled by $(w,\theta)$ is $\varepsilon$-regular if $w$ is. Similarly, the $\varepsilon$-critical index of $\phi$ is defined to be the $\varepsilon$-critical index of $w$.
\end{definition}

Define the parameter $L = L(\varepsilon) = (100\log^2(1/\varepsilon))/\varepsilon^2$.

The Berry--Esseen theorem (see, \eg, \cite{Feller}) yields the following standard anticoncentration lemma for linear functions. (See~\cite[Corollary 2.2]{DGJSV} for this particular statement.)

\begin{lemma}[Anticoncentration for regular linear functions]
\label{lem:anticon}
Let $w\in \mathbb{R}^n$ be $\varepsilon$-regular and let $J\subseteq \mathbb{R}$ be any interval. Then,
\[
\left|\prob{x\in \{-1,1\}^n}{\ip{w}{x}\in J} - \Phi(J)\right|\leq O(\varepsilon)
\] 
where $\Phi(\cdot)$ denotes the cdf of the standard Gaussian with mean $0$ and variance $\mynorm{w}_2^2$. In particular, if $|J|$ denotes the length of $J$, then
\[
\prob{x}{\ip{w}{x} \in J} \leq \frac{|J|}{\mynorm{w}_2} + O(\varepsilon).
\]
\end{lemma}

We now proceed with the proof of \expref{Lemma}{lem:thresh-gate}. Throughout, we work with random restrictions sampled from $\mc{R}_p^n$ where $p\in [0,1]$ is the probability from the statement of \expref{Lemma}{lem:thresh-gate}: equivalently, we pick a pair $(I,y)$ where $I\subseteq [n]$ and $y\in \{-1,1\}^{n-|I|}$. 

Let the threshold gate $\phi$ be labeled by pair $(w,\theta)$, where $w\in \mathbb{R}^n$. We may assume that the variables of the threshold gate have been sorted so that $|w_1|\geq |w_2|\geq \cdots \geq|w_n|$. Note that after applying a restriction $\rho$, the threshold gate $\phi|_\rho$ is labeled by pair $(w',\theta')$, where $w'$ is the restriction of $w$ to the coordinates in $I$ and
\begin{equation}
\label{eq:msl1}
\theta' = \theta'(\rho) = \theta - \ip{w''}{y}
\end{equation}
where we use $w''$ to denote the vector $w$ restricted to the indices in $[n]\setminus I$.

For a random restriction $\rho\sim \mc{R}_p^n$, define the following ``bad'' events: 
\begin{enumerate}
\item $\mc{B}^t(\rho)$ ($t$ a parameter): $\phi|_\rho$ is $t$-balanced: \ie, $\theta' \leq t\cdot \mynorm{w'}_2$. This is the event whose probability we want to upper bound.
\item $\mc{B}_1(\rho)$: $\sum_{i\in I} w_i^2 \geq 2 p\mynorm{w}_2^2$. 
\item $\mc{B}_2^{k,\ell}(\rho)$ ($k,\ell$ parameters): $I$ contains at least $k$ variables among the first $\ell$ variables $x_1,\ldots,x_\ell$.\footnote{We actually only need to analyze the events $\mc{B}_{2}^{1,\ell}$ for the proof of \expref{Lemma}{lem:thresh-gate}. However, we will need a slightly more general bound for the proof of \expref{Lemma}{lem:thresh-gate-2} later.}
\end{enumerate}

We have the following upper bounds on the probabilities of some of the above bad events.
\begin{claim}
\label{clm:B1}
Say $w$ is $\varepsilon'$-regular for any $\varepsilon' \leq \frac{1}{\sqrt{16\log (1/p)}}$. Then, $\prob{\rho}{\mc{B}_1(\rho)} \leq p^2$.
\end{claim}

\begin{proof}
  The proof follows by applying the variant of the Chernoff bound in \expref{Theorem}{thm:chunglu}. Since each variable is added to $I$ independently with probability $p$, we have
  \[
    \avg{\rho}{\sum_{i\in I} w_i^2} = p\cdot \mynorm{w}_2^2.
  \]
  Applying \expref{Theorem}{thm:chunglu} to the Boolean random
  variables $X_i$ which take value $1$ iff $x_i\in I$ and for
  $X = \sum_i w_i^2 X_i$, we obtain the following bound.
\[
\prob{}{\sum_{i\in I}w_i^2\geq 2\mynorm{w}_2^2} \leq \prob{}{X \geq \avg{}{X} + \mynorm{w}_2^2}\leq \exp(-\mynorm{w}_2^2/8\max_i\{w_i^2\})\leq p^2
\]
where for the last inequality we have used the fact that $w$ is $\varepsilon'$-regular and hence
\[
  \max_i w_i^2\leq (\varepsilon')^2\cdot \mynorm{w}_2^2 = \mynorm{w}_2^2/(16\log(1/p)).\qedhere
\]
\end{proof}

\begin{claim}
\label{clm:B2kell}
$\prob{\rho}{\mc{B}_2^{k,\ell}(\rho)}\leq (ep\ell/k)^k$.
\end{claim}

\begin{proof}
  The probability that any specific set of $k$ variables is a subset of $I$ is $p^k$. By a union bound, for any $\ell$, we have \[
    \prob{\rho}{\mc{B}_2^{k,\ell}(\rho)}\leq \binom{\ell}{k} p^k \leq (e\ell/k)^k \cdot p^k.\qedhere
  \]
\end{proof}

We start with a simpler subcase of the lemma that follows almost directly from \expref{Lemma}{lem:anticon}. 

\begin{lemma}[The regular case]
\label{lem:reg}
Say that $w$ is $\varepsilon'$-regular for some $\varepsilon'\leq \frac{1}{\sqrt{16\log(1/p)}}$. Then 
$$\prob{\rho}{\mc{B}^t(\rho)}\leq O(t\sqrt{p} + \varepsilon').$$
\end{lemma}

\begin{proof}
We bound $\prob{\rho}{\mc{B}^t(\rho)}$ as follows.
\begin{align}
\prob{\rho}{\mc{B}^t(\rho)} &\leq \prob{\rho}{\mc{B}_1(\rho)} + \prob{\rho}{\mc{B}^t(\rho)\ |\ \neg(\mc{B}_1(\rho))}\notag\\
&\leq p^2 + \prob{\rho}{\mc{B}^t(\rho)\ |\ \neg(\mc{B}_1(\rho))}\label{eq:reg1}
\end{align}
where the bound on $\prob{}{\mc{B}_1(\rho)}$ follows from \expref{Claim}{clm:B1}.

Now, note that the event $\neg\mc{B}_1(\rho)$ only depends on the choice of $\rho^{-1}(*) = I$. Hence we can condition on an $I$ so that this event occurs; choosing $\rho$ is now equivalent to choosing a random assignment $y$ to the variables in $[n]\setminus I$.

We have $\theta' = \theta - \ip{w''}{y}$. Using the fact that $\mc{B}_1(\rho)$ doesn't occur, we have 
\begin{itemize}
\item $\mynorm{w''}_2\geq \mynorm{w}_2\sqrt{1-2p}\geq \mynorm{w}_2/2$. Using the $\varepsilon'$-regularity of $w$, for each $i\not\in I$, we have
  \[
    |w_i|\leq \varepsilon'\mynorm{w}_2\leq 2\varepsilon'\mynorm{w''}_2.
  \]
  Thus, $w''$ is $2\varepsilon'$-regular.
\item $\mynorm{w'}_2 \leq 2p^{1/2}\mynorm{w}_2 \leq 4p^{1/2}\mynorm{w''}_2$,
\end{itemize}

Using the above, we can see that the probability that 
\begin{align*}
\prob{\rho}{\mc{B}^t(\rho)\ |\ \neg\mc{B}_1(\rho) } &\leq \prob{y}{|\theta'| \leq t\cdot \mynorm{w'}_2} \leq \prob{y}{|\theta'| \leq 4tp^{1/2}\cdot\mynorm{w''}_2}\\
&\leq \prob{y}{\ip{w''}{y}\in [\theta- 4tp^{1/2}\cdot\mynorm{w''}_2, \theta + 4tp^{1/2}\cdot\mynorm{w''}_2]}\\
&\leq O(tp^{1/2}) + O(\varepsilon') 
\end{align*}
where the final inequality uses the anti-concentration bound in \expref{Lemma}{lem:anticon}. Putting the above together with (\ref{eq:reg1}), we are done.
\end{proof}

\begin{proof}[Proof of \expref{Lemma}{lem:thresh-gate}]
The proof of the lemma is a standard case analysis based on the $\varepsilon$-critical index of the threshold gate $\phi$ (see~\cite{Ser,OS,DGJSV,MZ}). Let $\varepsilon = p^{1/8}$ and $t = p^{-1/16}$. The parameter $p_0$ will be chosen in the proof below. 

The first case is when the critical index $K\leq L$. In this case, we bound the probability of $\mc{B}^t(\rho)$ by 
\begin{align}
\prob{\rho}{\mc{B}^t(\rho)} &\leq \prob{\rho}{\mc{B}_2^{1,K}(\rho)} + \prob{\rho}{\mc{B}^t(\rho)\ |\ \neg \mc{B}_2^K(\rho)}\notag\\
&\leq epK + \prob{\rho}{\mc{B}^t(\rho)\ |\ \neg \mc{B}_2^K(\rho)} \leq \sqrt{p} + \prob{\rho}{\mc{B}^t(\rho)\ |\ \neg\mc{B}_2^K(\rho)}\label{eq:msl2}
\end{align}
where the second inequality follows from \expref{Claim}{clm:B2kell} and the final inequality follows from the fact that $epK \leq epL \leq \sqrt{p}$ by our choice of parameters. The event $\neg\mc{B}_2(\rho)$ only depends on the choice of the sub-restriction $\rho|_{[K]}$ and we can condition on $\rho|_{[K]}$ so that this event occurs. {}From now on, the random choice will be a restriction $\rho'\sim \mc{R}_p^{n-K}$ on the remaining variables.

Since the restricted linear function is now $\varepsilon$-regular by the definition of the $\varepsilon$-critical index, we can apply \expref{Lemma}{lem:reg} to conclude that
\[
  \prob{\rho'}{\mc{B}^t(\rho)\; \middle|\; \neg\mc{B}^{1,K}_2(\rho)}\leq p^{\Omega(1)}
\]
(note that $\varepsilon = p^{1/8} \leq 1/\sqrt{16\log(1/p)}$ as long as $p$ is smaller than some absolute constant $p_0$, so that \expref{Lemma}{lem:reg} is applicable). Along with (\ref{eq:msl1}), this implies the lemma in the case that $K\leq L$.

The second case is when $K > L$. As in previous cases, we first condition on some bad event not occurring. We have
\begin{align}
\prob{\rho}{\mc{B}^t(\rho)} &\leq \prob{\rho}{\mc{B}_2^{1,L}(\rho)} + \prob{\rho}{\mc{B}^t(\rho)\; \middle|\; \neg\mc{B}_2^{1,L}(\rho)}\notag\\
&\leq epL + \prob{\rho}{\mc{B}^t(\rho)\;\middle|\; \neg\mc{B}_2^{1,L}(\rho)} \leq \sqrt{p} + \prob{\rho}{\mc{B}^t(\rho)\;\middle|\; \neg\mc{B}_2^{1,L}(\rho)}.\label{eq:msl3}
\end{align}

As above, we can condition on a fixed $I$ so that $\mc{B}_2^{1,L}(\rho)$ does not occur (\ie, none of the first $L$ variables belong to $I$). We then use the following claim that is implicit in~\cite{DGJSV}.

\begin{proposition}
\label{prop:DGJSV}
Assume that $L' = (10r\log(1/\varepsilon))/\varepsilon^2$ and that the $\varepsilon$-critical index $K > L'$. Let $y$ be a random assignment to any set of variables including the first $L'$ variables. Then, the probability over $y$ that the restricted threshold gate is not $(1/\varepsilon)$-imbalanced is at most $2^{-r}$.
\end{proposition}

Applying the above proposition with $L' = L$ and $r = 10\log(1/\varepsilon)$, we have
\[
  \prob{\rho}{\mc{B}^t(\rho)\ |\ \neg\mc{B}_2^{1,L}(\rho)}\leq \varepsilon^{10}<p.
\]
Putting this together with (\ref{eq:msl3}), we have the claimed upper
bound on $\prob{\rho}{\mc{B}^t(\rho)}$ in the case that $K > L$.
\end{proof}

For completeness, we give below a proof sketch of \expref{Proposition}{prop:DGJSV}.

\begin{proof}[Proof of \expref{Proposition}{prop:DGJSV}]
Let $J$ be the set of variables being set and let $y\in \{-1,1\}^{|J|}$ denote the random assignment chosen. Let $L_0 = ({1}/{\varepsilon^2})\cdot 3\log(1/\varepsilon)$. It can be checked that for any $i < L' - L_0$, we have 
\[
\mynorm{w_{\geq (i+L_0)}}_2^2 \leq \frac{\varepsilon^2}{9}\cdot \mynorm{w_{\geq i}}^2\leq \frac{w_i^2}{9}.
\]

Hence, we can choose indices $i_1 = 1,i_2 = 1+L_0,\cdots,i_{r+2} = 1+(r+1)L_0\leq L'$ such that
\[
  |w_{i_{j+1}}| \leq \frac{|w_{i_j}|}{3}\qquad\text{and}\qquad \mynorm{w_{\geq i_{j+1}}}_2^2 \leq \frac{\varepsilon^2}{9}\cdot\mynorm{w_{\geq i_j}}_2^2.
\]
Further, we have
$$\sum_{i\not\in J}w_i^2\leq \mynorm{w_{\geq L'}}^2 \leq \mynorm{w_{\geq {i_{r+2}}}}^2\leq \frac{\varepsilon^2}{9}\cdot \mynorm{w_{\geq i_{r+1}}}_2^2\leq \frac{\varepsilon^2}{81}\cdot w_{i_r}^2.$$

We condition on any setting of variables other than $y_{i_1},\ldots,y_{i_r}$. This means that the constant term of the restricted threshold gate $\theta'$ is given by
\[
\theta' = \theta'' - \sum_{j\in [r]}w_{i_j}y_{i_j}
\]
for some $\theta'' \in \mathbb{R}$. The probability that the threshold gate is not ${1}/{\varepsilon}$-imbalanced is at most
\begin{align*}
\prob{y_{i_1},\ldots,y_{i_r}}{|\theta'|\leq \frac{1}{\varepsilon}\cdot \sqrt{\sum_{i\not\in J}w_i^2}} &\leq \prob{y_{i_1},\ldots,y_{i_r}}{|\theta'|\leq \frac{1}{9}\cdot |w_{i_r}|}\\
&= \prob{y_{i_1},\ldots,y_{i_r}}{\sum_{j=1}^r w_{i_j}y_{i_j} \in [\theta''- \frac{1}{9}\cdot |w_{i_r}|,\theta'' + \frac{1}{9}\cdot |w_{i_r}| ]}.
\end{align*}

Now, as a result of the exponentially decreasing nature of the $|w  _{i_j}|$, it follows that for any interval of length at most $|w_{i_r}|/2$, there can be at most one choice of $y_{i_1},\ldots,y_{i_r}$ such that the $\sum_j w_{i_j}y_{i_j}$ lies in that interval. Thus, we have the given bound.

\end{proof}

\section{Satisfiability algorithms beating brute-force search}\label{sec:satalgo}

In this section, we give satisfiability algorithms beating brute force search for bounded-depth threshold circuits with few wires. Until now, such algorithms were only known for threshold circuits of depth 2. We will assume that each threshold gate on $m$ input bits is given as a pair $(w, \theta)$, where $w \in \mathbb{Z}^m$ and $\theta \in \mathbb{Z}$, and $\theta$ as well as each component of $w$ has bit complexity $\poly(n)$. Note that this assumption is without loss of generality for a threshold function, and that some assumption on representability of threshold functions is necessary in an algorithmic context.

The satisfiability algorithm relies on an algorithmic version of \expref{Lemma}{lem:depthredn}, along with a couple of additional ideas. Essentially, we use the algorithmized version of the lemma to reduce the satisfiability of bounded-depth circuits to satisfiability of ANDs of threshold functions, which we can then solve using a recent result of Williams, stated below.

\begin{theorem} \cite{Will-ACC-THR}
\label{thm:AC0ThrSAT}
There is a deterministic algorithm, which given a bounded-depth circuit $C$ on $n$ variables of size $2^{n^{o(1)}}$ with ANDs, ORs and threshold gates, and with the threshold gates appearing only at the bottom layer, decides if $C$ is satisfiable in time $2^{n-n^{\varepsilon'}} \poly(n)$, where $\varepsilon' > 0$ is a constant that depends only on the depth of the circuit.
\end{theorem}

We also need the following fact about threshold gates on $n$ input bits: the set of inputs evaluating to $1$ (and dually, the set of inputs evaluating to $-1$) of a linear threshold gate can be enumerated in time proportional to the number of such inputs, modulo a $\poly(n)$ factor. 

\begin{proposition}
\label{prop:ThrEnum}
Let $(w, \theta)$ represent a threshold function $\phi$ on $m$ input bits, where $w \in \mathbb{Z}^m$ and $\theta \in \mathbb{Z}$ are integers of bit complexity $\poly(m)$.  Let $S$ be the set of inputs on which $\phi$ evaluates to 1. Then $S$ can be enumerated in time $|S| \poly(m)$.
\end{proposition}

\begin{proof}
We will show how to construct a decision tree for $\phi$ in time $|S| \poly(m)$, where $S$ is the set of inputs on which $\phi$ takes value 1. Given a decision tree of size at most $|S| \poly(m)$, it is easy to enumerate the set of inputs on which $\phi$ takes value 1 in time $|S| \poly(m)$ by scanning through leaves labeled 1 and outputting all assignments corresponding to any such leaf.

The decision tree is constructed recursively as follows. Check if $\phi$ restricted according to the current partial assignment is satisfiable (in the sense that there is a total assignment consistent with the partial assignment for which $\phi$ evaluates to 1). 
 Note that satisfiability of a linear threshold gate with polynomial bit complexity of the weights can be done trivially in polynomial time. If the satisfiability check fails, make the current node a leaf and label it with $-1$. If it succeeds, check if the current partial assignment is falsifiable. If this check fails, make the current node a leaf and label it with 1. Otherwise, branch on an arbitrary unassigned variable and recurse.
 
 Clearly, this decision tree can be constructed with polynomial work at each node, and hence in time $N \poly(m)$, where $N$ is the number of leaves of the tree. We show that $N \leq |S| m$. Indeed, we prove inductively that for any internal node $v$ of the tree of height $h \geq 1$, the number of $-1$ leaves of the tree rooted at $v$ is at most $h$ times the number of $1$ leaves, from which the claim follows as the height of the tree $\leq m$.
 
 For the inductive claim, the base case $h=1$ is clear as any node at height 1 must have one leaf labeled 1 and the other labeled $-1$. Assume the claim for height $h$. Consider a node $v$ at height $h+1$. Either one of its children is a leaf, or not. If one of the children is a leaf, then the other one $v'$ is not and by the induction hypothesis, since it is of height $h$, has at most $h$ times as many $-1$ leaves as 1 leaves. The number of $-1$ leaves of $v$ is at most one plus the number of $-1$ leaves of $v'$, and hence at most $h+1$ times the number of 1 leaves. In case both children of $v$ are internal nodes, then they are both of height at most $h$, and by the induction hypothesis, both have at most $h$ times as many $-1$ leaves as 1 leaves, which implies that the same holds for $v$.
\end{proof}

  \begin{definition}
 \label{def:thresholds}
 We use $\THR$ to refer to the class of linear threshold functions. We use $ \AND \circ \THR$ to refer to the class of polynomial-size circuits with an AND gate at the top and threshold gates at the bottom layer.
 \end{definition}

 \begin{theorem}
 \label{thm:sat-algorithm}
 For each integer $d > 0$, there is a constant $\varepsilon_d > 0$ such that satisfiability of a depth-$d$ threshold circuit with at most $n^{1+\varepsilon_d}$ wires on $n$ variables can be solved by a randomized algorithm in time $2^{n-\Omega(n^{\varepsilon_d})} \poly(n)$.
 \end{theorem}

\begin{proof}
As the proof follows the proof of \expref{Lemma}{lem:depthredn} closely, we just give a sketch. 
Call a circuit depth-$d$ $\AND \circ \THR$-skew if the top gate is an AND and all but one child of the top gate is a bottom-level threshold gate, with the possibly exceptional child being a depth-$(d-1)$ threshold circuit with few wires. We follow the depth reduction argument in the lemma to give a recursive algorithm which reduces satisfiability of polynomial-size depth-$d$ $\AND \circ \THR$-skew circuits to the satisfiability of polynomial-size depth-$(d-1)$ $\AND \circ \THR$-skew circuits by appropriately restricting variables.

For the base case $d=1$, we simply appeal to the algorithm given by \expref{Theorem}{thm:AC0ThrSAT}, which solves satisfiability of $\AND \circ \THR$ circuits of polynomial size in time $2^{n-n^{\varepsilon'}} \poly(n)$ for some constant $\varepsilon' > 0$. 

For the inductive case, we simulate the proof of \expref{Lemma}{lem:depthredn}, which performs and analyzes a certain kind of adaptive random restriction.  Various bad events might happen at Phases 2 and 3 of this random restriction process, however each step of the restriction process as well as the check that a bad event happens can be implemented in polynomial time. Moreover, the probability that a bad event happens is at most $2^{-n^{\varepsilon_d}}$. Whenever a bad event happens, we simply do brute force search on the remaining variables of the circuit, but thanks to the exponentially small probability that a bad event happens, with high probability, we only spend time $2^{n-n^{\varepsilon_d}}$ on such brute force searches.

In Phase 3 of the restriction process, we replace imbalanced gates by their most probable values. This changes the functionality of the circuit and might lose us satisfying assignments or give us new invalid satisfying assignments. To get around this, for each such imbalanced gate, we use \expref{Proposition}{prop:ThrEnum} to efficiently enumerate the inputs evaluating to the minority value for each imbalanced gate, and for each such input check whether it satisfies the original circuit. If it does, we just output ``yes.'' We also append to the top gate of the skew circuit a child representing the assignment of the imbalanced gate to its majority value---this needs to be done so that we don't end up with ``false positives'' in the base case of the recursive algorithm. Although each such false positive can be tested, there might be too many of them, and this could destroy all the savings we accrue through the course of the algorithm. The total time spent in enumerating minority values of imbalanced gates is again at most $2^{n-n^{\varepsilon_d}} \poly(n)$, with high probability, using the efficient enumeration and the imbalance property.

Finally, there are a few balanced gates---with high probability at most $O(n^{\delta_d})$ of them---for which we need to try all possible values. This could be expensive, but is compensated for by an increased savings for depth $d-1$, just by setting the constant $B$ large enough in the proof of \expref{Lemma}{lem:depthredn}. We also need to set $B$ large enough so that the savings given by the application of Williams' algorithm in the base case overwhelms the loss due to branching on balanced threshold gates at depth $d=2$.

Thus the total running time, once $B$ is chosen appropriately, is $2^{n-\Omega(n^{\varepsilon_d})} \poly(n)$, using the fact that $\varepsilon_d < \varepsilon_{d-1} < \ldots < \varepsilon_{2}$.
\end{proof}

\section{Threshold formulas}
\label{sec:formulas}

A threshold formula is a threshold circuit such that the fan-out of each gate is at most 1.
A formula can be viewed as a rooted tree.
Note that a depth-2 threshold circuit can always be converted to a threshold formula without increasing either the wire complexity or the gate complexity (recall that the gate complexity only measures the number of \emph{non-input} gates).

\begin{theorem}\label{thm:formula-desc}
There are  $2^{O(sn)}$ distinct Boolean functions computed by threshold formulas with $s$ wires on $n$ inputs.
\end{theorem}
\begin{proof}
A formula with $s$ wires is a rooted tree with $s+1$ nodes.
The number of different rooted trees with $s+1$ nodes is $2^{O(s)}$~\cite{Ott48}.
For a fixed tree structure,
we label the leaves by variables $x_1,\ldots, x_n$, and label the internal nodes by LTFs.
Since the number of leaves is at most $s$,
there are $O(n^s)$ different ways of labeling the leaves.

We next label the internal nodes. 
Without loss of generality, we assume leaves feeding into the same node are labeled by distinct variables,
and each internal nodes has fan-in at least two.
Consider a topological order of the internal nodes, say, $h_1,\ldots, h_k$, where $k \leq s$.
For $i=1,\ldots, k$, let $s_i$ be the fan-in of $h_i$; then $\sum_i s_i = s$.
We label each $h_i$ by a LTF on $s_i$ inputs.
For fixed labels of $h_j$ for $1 \leq j<i$, 
by \expref{Theorem}{thm:RSO94},
there are  $2^{O(ns_i)}$ possible choices for $h_i$.

Therefore, the number of different functions computed by threshold formulas of $s$ wires is
\[
2^{O(s)} \cdot O(n^s) \cdot 2^{O(ns_1)} \cdots 2^{O(ns_k)}  \leq  2^{O(ns)}. \qedhere
\]
\end{proof}

The main result of this section is the following.

\begin{theorem}~\label{thm:formula-corr}
Let $\gamma > 0$ be any constant parameter. 
Any threshold formula on $5n$ variables with at most $n^{2-3\gamma}$ wires  
 has correlation at most 
$\exp( - {n^{\Omega(\gamma)}})$ with the Generalized Andreev function as defined in \expref{Section}{sec:andreev} with parameter $\gamma$.
\end{theorem}

\begin{proof}
Let $C$ be a threshold formula with $n$ inputs and $s = n^{2 - 3\gamma}$ wires.
Let $L$ be the number of leaves in the formula tree; then $L\leq s \leq 2L$.
We build a restriction tree $T$ for $C$ up to depth $n-pn$, for $p = n^{\gamma }  /  n$, 
by greedily restricting the most frequent variables appearing in the leaves.
Since the most frequent variable appears at least $L/n$ times,
after restricting one variable, the formula tree has at most $L (1 - 1/n)$ leaves left.
Continue this until $pn$ variables left unrestricted; 
then the number of remaining leaves is at most
\[
  L  \cdot \frac{n-1}{n}  \cdot\frac{n-2}{n-1}\cdot \cdots \cdot \frac{pn}{pn+1} = pL.
\]
Thus, for any leaf $l$ of $T$,  the restricted formula $C|_{\rho_l}$, on $n' :=pn = n^{\gamma} $ variables,
has
\[
  s' := s(C|_{\rho_l})  \leq 2pL \leq 2 p s
\]
wires.  By \expref{Theorem}{thm:formula-desc}, there are at most
$2^{O(n's')}$ different threshold formulas of $s'$ wires on $n'$
variables, and thus each such formula has a description of length
$O(n's') = O(p^2 sn) = O(n^{1 - \gamma}) \ll n$.

Recall that $F \colon \{-1,1\}^{4n} \times  \{-1,1\}^{n}  \to  \{-1,1\}$ is the Generalized Andreev function defined in \expref{Section}{sec:andreev}.
Let $a \in \{-1,1\}^{4n}$ be a string with Kolmogorov complexity $K(a) \geq 3n$,
and let $F_a(x)  := F(a, x)$.
Then, by \expref{Lemma}{lem:ckksz-corr},
\[
  \Corr(F_a |_{\rho_l}, C|_{\rho_l}) \leq \exp( - {n^{\Omega(\gamma)}}).
\]
Note that this holds for every leaf $l$ of $T$.
By \expref{Fact}{fac:corr},
$\Corr(F_a, C) \leq \exp( - {n^{\Omega(\gamma)}})$.

Let $D$  be a threshold formula with $5n$ inputs and $n^{2 -3\gamma}$ wires,
and let $D_a(x)  := D(a, x)$.
Then $D_a$ is a formula on $n$ inputs with at most  $n^{2 -3\gamma}$ wires,
and thus
\[
  \prob{x}{ F_a(x)  =  D_a(x) } \leq 1/2 + \exp( - {n^{\Omega(\gamma)}}).
\]
Since a random $a \in \{-1,1\}^{4n}$ has $K(a) \geq 3n$ with probability $1 - 2^{-\Omega(n)}$,
the correlation of $D$ and $F$ is at most 
$2^{-\Omega(n)} + \exp( - {n^{\Omega(\gamma)}})  = \exp( - {n^{\Omega(\gamma)}})$.
\end{proof}

\section{\texorpdfstring{$\AC^0$}{AC0} circuits with a few threshold gates}
\label{sec:tac0r}
\newcommand{\OR}{\mathrm{OR}}

In this section, we extend the noise sensitivity and correlation bounds from \expref{Section}{sec:gates-basic} to the more general setting of small $\AC^0$ circuits (\ie, a constant-depth circuit made up of $\AND$ and $\OR$ gates) augmented with a small number of threshold gates. 

We prove noise sensitivity bounds for Boolean functions computed by such circuits. As consequences of this, we are able to prove correlation bounds against such circuits (as in \expref{Section}{sec:gates-basic}) and also devise learning algorithms for such circuits under the uniform distribution (as in~\cite{KOS,GS}). 

Following Gopalan and Servedio~\cite{GS}, we define $\TAC^0[k]$ to be the class of constant-depth circuits made up of $\AND$ and $\OR$ gates and at most $k$ arbitrary threshold gates. The inputs to the circuit are allowed to be arbitrary literals over the underlying variables. 

We prove upper bounds on the noise sensitivity of small depth-$d$ $\TAC^0[k]$ circuits for $k$ much smaller than $n^{1/2(d-1)}$. The main theorem of the section is the following.

\begin{theorem}
\label{thm:ns-taco}
Fix any constant $d\geq 1$. Let $C$ be a depth-$d$ $\TAC^0[k]$ circuit with at most $M$ gates overall. Then, for any $p,q\in [0,1]$ and any $D\geq 1$, we have
\[
\NS_{p^{d-1}q}(C)\leq O(k\alpha(p,D) + \alpha(q,D) + M(10pD)^D)
\]
where $\alpha(p,D):= \sqrt{p}\cdot (\log(1/p))^{O(D\log D)}\cdot 2^{O(D^2\log D)}$ and $O(\cdot)$ hides an absolute constant (independent of $p$ and $D$).
\end{theorem}

This implies a correlation bound for the parity function for such circuits as in \expref{Section}{sec:gates-basic} (see \expref{Corollary}{cor:taco} below). Using \expref{Theorem}{thm:ns-taco} along with a general idea due to Klivans et al.~\cite{KOS}, we also get the following subexponential-time (\ie, $2^{o(n)}$-time) learning algorithms for $\TAC^0[k]$ circuits of small size.

\begin{theorem}
\label{thm:learning-taco}
Let $d$ be any fixed constant. The class of $\TAC^0[k]$ circuits of depth $d$ and size $M$ where
\[
  M = n^{o(\sqrt{\log n/\log \log n})} \qquad\text{and}\qquad k = \delta n^{1/2(d-1)}
\]
for some $\delta > 0$ can be learned to within error $\varepsilon > 0$ in time $n^{O(m)}$ where
\[
  m = \max\{k^{2(d-1)}/\varepsilon^{2d},n^{1/4+o(1)}/\varepsilon^2\}.
\]
In particular, if $\delta = n^{-\Omega(1)}$ and
$\varepsilon = \Omega(1)$, then the running time of the algorithm is
$2^{o(n)}$.
\end{theorem}

\subsection{Noise sensitivity bounds for \texorpdfstring{$\TAC^0[k]$}{TAC0[k]} circuits (\expref{Theorem}{thm:ns-taco})}
We now prove \expref{Theorem}{thm:ns-taco}. The basic idea is the same as in \expref{Theorem}{thm:ns-gates}, but we also need to use the following powerful result of Kane~\cite[Corollary 3]{Kane}. 

\begin{lemma}[Kane~\cite{Kane}]
\label{lem:kane-ns}
Let $f$ be a degree-$D$ Polynomial Threshold function (PTF).\footnote{We refer the reader to \expref{Section}{sec:defnsTF} for the definitions of Polynomial Threshold functions.} Then, for any $p > 0$, $$\NS_p(f) \leq \sqrt{p}\cdot (\log(1/p))^{O(D\log D)}\cdot 2^{O(D^2\log D)}.$$
\end{lemma}

\begin{proof}[Proof of \expref{Theorem}{thm:ns-taco}]
This is a standard switching argument (see, \eg,~\cite{Hastad}) augmented with the ideas of \expref{Theorem}{thm:ns-gates}. We assume throughout that $q\leq 1/2$ without loss of generality since otherwise $\alpha(q,D)\geq q \geq {1}/{2}$ and the claim is trivial. 

We say that a threshold gate is a \emph{true} threshold gate if it is not an AND or OR gate.

For any parameters $k_1,d_1,t_1,s_1\in\mathbb{N}$ with $d_1 \geq 2$, we define $\TAC^0[k_1,d_1,t_1,s_1]$ to be the class of constant-depth circuits made up of $\AND,\OR$ and threshold gates such that:
\begin{itemize}
\item the overall depth is at most $d_1$,
\item the total number of gates at \emph{depth at most $d_1-2$} in the circuit is at most $s_1$,
\item all the true threshold gates  are at depth at most $d_1-2$ and there are at most $k_1$ of them, and
\item the bottom fan-in of the circuit (\ie, the maximum fan-in of a gate at depth $d_1-1$) is at most $t_1$.
\end{itemize}

Note that the circuit $C$ in the statement of the theorem is in the class $\TAC^0[k,d+1,1,M]$, since we may replace the input literals with (say) $\AND$ gates of fan-in $1$ at the expense of increasing the depth by $1$ but in the process satisfying all the criteria of the above definition. We prove the following stronger statement: for any $p,q,D$ as in the statement of the theorem, and any $C$ from the class $\TAC^0[k,d,D,M]$ with $d\geq 2$, we have
\begin{equation}
\label{eq:taco1}
\NS_{p^{d-2}q}(C) = \avg{\rho_d}{\Var(C|_{\rho_d})}\leq O(k\alpha(p,D) + \alpha(q,D) + M(10pD)^D)
\end{equation}
where $\rho_d\sim \mc{R}_{p_d}^n$ and $p_d := 2p^{d-2}q\in [0,1]$. Proving (\ref{eq:taco1}) will clearly prove the theorem.

The proof is by induction on $d$. The base case is $d=2$. In this case, since there are no true threshold gates at depth $d-1$ by assumption, a true threshold gate  can only occur as the output gate of the circuit $C$. Since $\AND$ and $\OR$ gates are also threshold gates, we can assume that the output gate is a threshold gate. The bottom fan-in being at most $D$ implies that each gate at depth $1$ can be represented exactly as a polynomial of degree at most $D$ and therefore that the function computed by $C$ is a degree-$D$ PTF\@. Hence, \expref{Lemma}{lem:kane-ns} trivially implies the result.

Now assume $d > 2$. Let $\psi_1,\ldots,\psi_s$ denote the AND and OR gates at depth exactly $d-2$ in the circuit and let $\phi_1,\ldots,\phi_m$ denote the true threshold gates. By assumption $m\leq k$ and $s\leq M$. We sample a random restriction $\rho\sim \mc{R}_p^n$ and consider the restricted circuit $C|_\rho$.

H\r{a}stad's switching lemma~\cite{Hastad} tells us that for each $i\in [s]$, we have
\begin{equation}
\label{eq:taco2}
\prob{\rho}{\text{DT-depth}(\psi_i|_\rho)\geq D}\leq (10pD)^D,
\end{equation}
and hence by a union bound,
\begin{equation}
\label{eq:taco3}
\prob{\rho}{\exists i\in [s]: \text{DT-depth}(\psi_i|_\rho)\geq D}\leq s(10pD)^D.
\end{equation}

Also, as in the base case, we see that each $\phi_j$ computes a degree-$D$ PTF\@. Hence, \expref{Lemma}{lem:kane-ns} gives us
\begin{equation}
\label{eq:taco4}
\avg{\rho}{\sum_{j\in [m]}\Var(\phi_j|_\rho)}\leq m\alpha(p,D).
\end{equation}

Consider the circuit $C'_\rho$ obtained from $C|_\rho$ as follows: if there is an $i\in [s]$ such that $\text{DT-depth}(\psi_i|_\rho)\geq D$, then $C'_\rho$ is defined to be a trivial circuit that always outputs $1$; otherwise, $C'_\rho$ is the depth-$(d-1)$ circuit obtained from $C|_\rho$ as follows:
\begin{itemize}
\item We replace each $\phi_j|_\rho$ by a bit $b_{j,\rho}\in \{-1,1\}$ so that by \expref{Fact}{fac:var}, we have
\[
\prob{x\in \{-1,1\}^{|\rho^{-1}(*)|}}{\phi_j(x) \neq b_{j,\rho}} \leq O(\Var(\phi_j)),
\]
\item Since each $\psi_i|_\rho$ is a depth-$D$ decision tree, we can write it as a $D$-DNF or $D$-CNF or as a \emph{disjoint sum of terms} of size at most $D$ each. For each gate $\chi$ at depth at most $d-3$ that takes $\psi_i$ as an input, we do the following:
\begin{itemize}
\item If $\chi$ is an OR gate, then we take the $D$-DNF representing $\psi_i|_\rho$ and feed the terms of the DNF directly into $\chi$, eliminating the output OR gate of the $D$-DNF.
\item If $\chi$ is an AND gate, we do the same as above, except that we use the $D$-CNF representation of $\psi_i|_\rho$ and eliminate the output AND gate.
\item If $\chi$ is a threshold gate, then we write $\psi_i|_\rho$ as a disjoint sum of terms of size at most $D$ each and feed each of the terms directly to $\chi$. The gate $\chi$ now has many inputs in the place of $\psi_i|_\rho$, and the weight given to each of these inputs is the same as the weight given to $\psi_i|_\rho$.
\end{itemize}
Note that the above operations do not increase the number of gates at depth at most $d-3$ in the circuit.
\end{itemize}

Note that $C'_\rho$ has depth $d-1$ and bottom fan-in at most $D$. Further, the number of gates at depth at most $d-3$ in $C'_\rho$ is at most $M-s$. Hence, $C'_\rho$ is a circuit from the class $\TAC^0[k-m,d-1,D,M]$. We can thus apply the induction hypothesis and obtain
\begin{equation}
\label{eq:taco5}
\avg{\rho_{d-1}}{\Var(C'_\rho|_{\rho_{d-1}})}\leq O((k-m)\alpha(p,D) + \alpha(q,D) + (M-s)(10pD)^D).
\end{equation}

To obtain~(\ref{eq:taco1}), we use 
\begin{align}
\avg{\rho_d}{\Var(C)} = \avg{\rho_{d-1}}{\avg{\rho}{\Var(C|_\rho)|_{\rho_{d-1}}}} &\leq \avg{\rho_{d-1}}{\avg{\rho}{\Var(C'_\rho)|_{\rho_{d-1}}}} + O\left(\avg{\rho}{\delta(C|_\rho, C'_\rho)}\right)\notag\\
& = \avg{\rho}{\avg{\rho_{d-1}}{\Var(C'_\rho)|_{\rho_{d-1}}}} + O\left(\avg{\rho}{\delta(C|_\rho, C'_\rho)}\right)\label{eq:taco6}
\end{align}
 where the inequality follows from \expref{Proposition}{prop:var}. Inequality (\ref{eq:taco5}) allows us to bound the first term on the right hand size. 

 It remains to analyze the last term on the right hand side of (\ref{eq:taco6}). Define a Boolean random variable $Z = Z(\rho)$ which is $1$ iff there is an $i\in [s]$ such that $\phi_i$ is not a depth-$D$ decision tree. Let $\Delta = \Delta(\rho)$ be the random variable defined by
 \[
   \Delta := Z + \sum_{j\in [m]}\Var(\phi_j|_\rho).
 \]

It easily follows from the definition of $C'_\rho$ that for any choice of $\rho$, either $Z=1$---in which case we can trivially bound $\delta(C'_\rho,C|_\rho)$ by $1$---or $$\delta(C'_\rho,C|_\rho)\leq \sum_{j}\delta(\phi_j|_\rho,b_{j,\rho}) = \sum_j \prob{x}{\phi_j|_\rho(x)\neq b_{j,\rho}}.$$ Hence, for any choice of $\rho$, we get
\[
\delta(C'_\rho,C|_\rho) \leq Z + \sum_{j\in [m]}\prob{x\in \{-1,1\}^{|\rho^{-1}(*)|}}{\phi_j|_{\rho}(x)\neq b_{j,\rho}} \leq O(\Delta).
\]

Further, by (\ref{eq:taco3}) and (\ref{eq:taco4}), we have
\[
  \avg{\rho}{\Delta}\leq O(m\alpha(p,D) + s(10pD)^D).
\]
Putting this together with (\ref{eq:taco5}) and (\ref{eq:taco6})\footnote{Of course, we need to be judicious in our choice of constants in the $O(\cdot)$. We leave this matter to the interested reader.} gives the claimed bound. This completes the induction.
\end{proof}

This yields the following correlation bound as in \expref{Corollary}{cor:gates}.

\begin{corollary}
\label{cor:taco}
The following is true for any constant $d\geq 2$. Say $C$ is a depth-$d$ $\TAC^0[k]$ circuit with at most $M$ gates where
\[
  k\leq \delta\cdot n^{1/2(d-1)}\qquad\text{and}\qquad M = n^{o(\sqrt{\log n/\log \log n})}.
\]
Then $\Corr(C,\Par_n)\leq n^{o(1)}\cdot \delta^{1-(1/d)}$. In
particular, if $\delta = n^{-\Omega(1)}$, then
$\Corr(C,\Par_n) = n^{-\Omega(1)}$.
\end{corollary}

\begin{proof}
We choose a $D$ such that $\omega(1)\leq D\leq o(\sqrt{\log n/\log \log n})$ so that $M\leq n^{o(D)}$ and $p,q$ as in \expref{Corollary}{cor:gates}. We can then use \expref{Theorem}{thm:ns-taco} to obtain $$\NS_{1/n}(C)\leq n^{o(1)}\cdot \delta^{1-(1/d)} + M\cdot (10pD)^D\leq 1/n^{\omega(1)}.$$ 
Thus, we get $\NS_{1/n}(C)\leq n^{o(1)}\delta^{1-({1}/{d})}$. By \expref{Proposition}{prop:ns}, we have $\Corr(C,\Par_n)\leq O(\NS_{1/n}(C))$, which proves the claim.
\end{proof}

\begin{remark}
  The above corollary can be strengthened considerably if a widely believed strengthening of \expref{Lemma}{lem:kane-ns}---named the Gotsman-Linial conjecture~\cite{GL}---is known to hold. The Gotsman-Linial conjecture is a conjecture about the \emph{average sensitivity} of low-degree PTFs. We do not recall the exact statement of the conjecture here, and refer the reader to the
article by  %
Gopalan and Servedio~\cite{GS} instead. As noted by~\cite[Corollary 13]{GS}, the Gotsman-Linial conjecture implies that for any $p$ and any degree $D$ PTF, we have $\NS_p(f)\leq O(D\sqrt{p})$. Plugging in this bound in place of Lemma 45, it is not hard to see that we can obtain $\Corr(C,\Par_n) = o(1)$ for any circuit $C$ of size $2^{n^{o(1)}}$ from the class $\TAC^0[k]$ where $k = n^{1/2(d-1)-\Omega(1)}$. This is almost a complete generalization of the result of Beigel~\cite{Beigel} who proved such a result in the setting where all the threshold gates are of polynomial weight. In contrast, the results of Podolskii~\cite{Podolskii} and Gopalan and Servedio~\cite{GS} can prove such correlation bounds only if $k < \log n$.
\end{remark}

\subsection{Learning algorithms for \texorpdfstring{$\TAC^0[k]$}{TAC0[k]} circuits (\expref{Theorem}{thm:learning-taco})}

To prove \expref{Theorem}{thm:learning-taco}, we use \expref{Theorem}{thm:ns-taco} along with an observation of Klivans, O'Donnell, and Servedio~\cite{KOS}. We have the following lemma that can be obtained by putting together Fact 9 and Corollary 15 in~\cite{KOS}.

\begin{lemma}
\label{lem:KOS}
Let $\mc{F}$ be a class of Boolean functions defined on $\{-1,1\}^n$. Assume that we know that for some $\varepsilon > 0$ and $f\in \mc{F}$, there is a $\gamma > 0$ such that $\NS_\gamma(f)\leq \varepsilon/3$. Then, there is an algorithm that learns $\mc{F}$ with error $\varepsilon$ in time $n^{O(1/\gamma)}$.
\end{lemma}

We now prove \expref{Theorem}{thm:learning-taco}.

\begin{proof}[Proof of \expref{Theorem}{thm:learning-taco}]
We can assume that $\varepsilon \geq 1/n^{1/2d}$ since otherwise, we can just run a brute force algorithm that takes time $2^{O(n)}$. We choose a $D$ such that $\omega(1)\leq D\leq o(\sqrt{\log n/\log \log n})$ so that $M= n^{o(D)}$. \expref{Theorem}{thm:ns-taco} tells us that for any $p,q\geq {1}/{n}$ any $C$ from the class of circuits described in the theorem statement, we have
\[
\NS_{p^{d-1}q}(C) \leq Ak\sqrt{p} + B\sqrt{q} + O(M(10pD)^D)
\]
where $A$ and $B$ are $n^{o(1)}$. 

We choose $p,q$ so that the first two terms above are each bounded by $\varepsilon/10$. This requires $p \leq \varepsilon^2/O(k^2 A^2)$ and $q \leq \varepsilon^2/O(B^2)$. Further, to ensure that the last term is at most $\varepsilon/10$, it suffices to choose $p \leq n^{-\Omega(1)}$ (in fact, this ensures that the third term is $n^{-\omega(1)}$ whereas $\varepsilon \geq n^{-1/2d}$ by assumption). Thus, we fix $p = \min\{\varepsilon^2/O(k^2 A^2),n^{-1/4d}\}$ and $q = \varepsilon^2/O(B^2)$ so that all the above conditions are satisfied. This gives 
\[
\NS_\gamma(C) \leq \varepsilon/3
\]
where $\gamma = p^{d-1}q$. Hence, by \expref{Lemma}{lem:KOS}, we obtain the statement of the theorem.
\end{proof}

\begin{remark}
\label{rem:learning}
Assuming the Gotsman Linial conjecture, the above technique yields subexponential time constant-error learning algorithms as long as $M \leq 2^{n^{o(1)}}$ and $\delta = n^{-\Omega(1)}$. To contrast again with the work of Gopalan and Servedio~\cite{GS}, the results of~\cite{GS}---even assuming the Gotsman Linial conjecture---only yield
subexponential-time    %
learning algorithms in the setting when $k < \log n$. However, the dependence on the error parameter in~\cite{GS} is better than the dependence we obtain
here.  (The  %
running time there has a $\varepsilon^3$ in place of the $\varepsilon^{2d}$
that we obtain
here.)      %
\end{remark}

\section{Depth-3 threshold circuit lower bounds}
\label{sec:kw}

\newcommand{\MAJ}{\mathrm{MAJ}}
\newcommand{\TC}{\mathrm{TC}}

In this section, we compare our lower bound techniques for threshold functions with the
simultaneous  %
independent work of Kane and Williams~\cite{KW}, which proves improved lower bounds for depth-$3$ threshold circuits. While both results analyze the effects of random restrictions on threshold gates, the random restriction lemma obtained in~\cite{KW} is  different from ours. It thus makes sense to ask if our techniques can also be used to prove the lower bounds of~\cite{KW}. We show that a modified version of our structural lemma (\expref{Lemma}{lem:thresh-gate}) can be used to recover the lower bound of~\cite{KW} up to $n^{o(1)}$ factors.

The modified structural lemma is as follows. We say that a threshold gate $\phi$ with label $(w,\theta)$ is \emph{$(t,k)$-imbalanced} if there is a set $S$ of at most $k$ input  variables to $\phi$ such that for each setting to these  variables, the resulting threshold gate is $t$-imbalanced. We prove the following modification of \expref{Lemma}{lem:thresh-gate}.

\begin{lemma}
\label{lem:thresh-gate-2}
Let $\phi$ be a threshold gate on $n$ variables with label $(w,\theta)$. For a parameter $p\in (0,1)$, and any $t,k\in \mathbb{N}$ such that $t = (1/p)^{o(1)}$ and $k = \omega(1)$ (\ie, $\log_{1/p}(t)\rightarrow 0$ and $k\rightarrow\infty$ as $p\rightarrow 0$), we have
\[
\prob{\rho\in \mc{R}_p^n}{\text{$\phi|_\rho$ is not $(t,k)$-imbalanced}} \leq p^{1/2-o(1)}.
\]
(Here, upper bounding the probability of an event by $p^{1/2-o(1)}$ means that the probability is at most $\sqrt{p}\cdot g(p)$ where $\lim_{p\rightarrow 0}g(p)p^{\delta} = 0$ for any fixed $\delta > 0$.)
\end{lemma}

\begin{remark}
\label{rem:thr-gate-2}
The statement of \expref{Lemma}{lem:thresh-gate-2} is somewhat incomparable to that of \expref{Lemma}{lem:thresh-gate}. On the one hand, \expref{Lemma}{lem:thresh-gate} yields a stronger property: namely that the threshold gate is $(t,0)$-imbalanced (or equivalently, not $t$-balanced) with high probability; \expref{Lemma}{lem:thresh-gate-2} only yields that the gate is $(t,k)$-imbalanced with high probability (and also holds only for a smaller value of $t$). However, the probability estimates in \expref{Lemma}{lem:thresh-gate-2} are stronger than those in \expref{Lemma}{lem:thresh-gate}: it can be checked that the proof of \expref{Lemma}{lem:thresh-gate} cannot yield an upper bound of $p^c$ for some $c < 1/2$.
\end{remark}

Replacing the random restriction lemma of~\cite{KW} with \expref{Lemma}{lem:thresh-gate-2} in the lower bound proof of~\cite{KW} yields the same lower bound for depth-$3$ circuits up to $n^{o(1)}$ factors. We state the result below after a few basic definitions.

An \emph{unweighted majority} gate is a threshold gate whose label $(w,\theta)$ satisfies $|w_i| = 1$ for all $i$. We let $\MAJ\circ \THR\circ\THR$ denote the class of depth-$3$ circuits that have an unweighted majority gate as the output gate above two layers of general threshold gates.

\begin{theorem}[Kane-Williams~\cite{KW}]
\label{thm:KW}
For infinitely many $n\in\mathbb{N}$, there is an explicit (polynomial-time computable) function  defined on $2n$ variables that does not have $\MAJ\circ\THR\circ\THR$ circuits with fewer than $n^{1.5-o(1)}$ gates or $n^{2.5-o(1)}$ wires.
\end{theorem}

\subsection{A modified structural lemma}

In this section, we prove \expref{Lemma}{lem:thresh-gate-2}. We will need the following easy consequence of \expref{Theorem}{thm:chernoff}.

\begin{fact}
\label{fact:tkimbal}
Let $\phi$ be a $(t,k)$-imbalanced threshold gate such that for some set $S$ of at most $k$ input  variables to $\phi$ and for each setting to these variables, the resulting threshold gate is $t$-imbalanced. The threshold gate $\phi'$ obtained from $\phi$ by keeping only the variables in $S$ (with their weights) and removing the others satisfies
\[
  \prob{x}{\phi(x)\neq \phi'(x)}\leq \exp(-\Omega(t^2)).
\]
\end{fact}

We now prove \expref{Lemma}{lem:thresh-gate-2}.

\begin{proof}[Proof of \expref{Lemma}{lem:thresh-gate-2}]
The proof of the above lemma will follow along the lines of the proof of \expref{Lemma}{lem:thresh-gate}. We also use the definitions and notation from \expref{Section}{subsec:msl}.

Throughout, we assume that $p$ is a small enough constant since the statement of the lemma only needs to be proved for $p\rightarrow 0$. We can also assume that $k = o(\lg(1/p))$ since assuming an upper bound on $k$ only makes the statement stronger. We sample $\rho = (I,y)$ from $\mc{R}_p^n$. 

We will need a number of parameters in the proof. Let $\varepsilon = m\sqrt{p}$ for an integer parameter $m \leq (1/p)^{o(1)}$ to be chosen later.  Define $q = 16m^2p\lg(1/p)$. Note that $q = p^{1-o(1)}$ since $m \leq (1/p)^{o(1)}$. Let  $L = 100\log^2(1/\varepsilon)/\varepsilon^2$. At some places in the proof, we will assume that $p$ is smaller than some constants, since the statement of the lemma is only non-trivial for $p\rightarrow 0$.

We assume that the variables of the threshold gate have been sorted so that $|w_1|\geq |w_2|\geq \cdots \geq |w_n|$. After applying a restriction $\rho$, the threshold gate $\phi|_\rho$ is labeled by pair $(w',\theta')$, where $w'$ is the restriction of $w$ to the coordinates in $I$ and
\begin{equation}
\label{eq:smsl1}
\theta' = \theta'(\rho) = \theta - \ip{w''}{y}.
\end{equation}
where $w''$ denotes the vector $w$ restricted to the indices in $[n]\setminus I$.

For a random restriction $\rho\sim \mc{R}_p^n$, we define the bad events $\mc{B}^t(\rho),\mc{B}_1(\rho)$ and $\mc{B}^{k,\ell}_2(\rho)$ exactly as in the proof of \expref{Lemma}{lem:thresh-gate} (here $\ell$ is a varying parameter we will choose below and $k$ is as fixed above). Additionally, we define the bad event $\tilde{\mc{B}}^t(\rho)$ to be the event that $\phi$ is not $(t,k)$-imbalanced: note that this is exactly the event whose probability we need to upper bound. Also note that $\mc{B}^t(\rho)$ holds whenever $\tilde{\mc{B}}^t(\rho)$ does and hence any upper bound on the probability of $\mc{B}^t(\rho)$ also applies to $\tilde{\mc{B}}^t(\rho)$.

Exactly as in \expref{Lemma}{lem:thresh-gate}, we first the case that $w$ is $\varepsilon$-regular and bound $\prob{\rho}{\tilde{\mc{B}}^t(\rho)}$ in this case. This easily follows since $\prob{}{\tilde{\mc{B}}^t(\rho)}\leq \prob{}{\mc{B}^t(\rho)}$ and hence by \expref{Lemma}{lem:reg} we have 
\begin{align}
\prob{\rho}{\tilde{\mc{B}}^t(\rho)}\leq \prob{\rho}{\mc{B}^t(\rho)} &\leq O(t\sqrt{p}+\varepsilon) = p^{1/2-o(1)}\label{eq:smsl2}
\end{align}
where we have used the fact that $t = (1/p)^{o(1)}$ and $\varepsilon = m\sqrt{p} = p^{1/2-o(1)}$ since $m = (1/p)^{o(1)}$.

In the non-regular case, we proceed by case analysis based on the $\varepsilon$-critical index $K$ of $w$. Specifically, we proceed based on whether $K \leq L$ or not.

We first show how to handle the former case, \ie, $K \leq L$. In this case, we bound the probability of $\tilde{\mc{B}}^t(\rho)$ as follows
\begin{align}
\prob{\rho}{\tilde{\mc{B}}^t(\rho)} &\leq \prob{\rho}{\mc{B}_2^{k,K}(\rho)} + \prob{\rho}{\mc{B}^t(\rho)\; \middle|\; \neg \mc{B}_2^{k,K}(\rho)}\notag\\
&\leq (epK/k)^k + \prob{\rho}{\mc{B}^t(\rho)\; \middle|\; \neg \mc{B}_2^{k,K}(\rho)} \leq \sqrt{p} + \prob{\rho}{\mc{B}^t(\rho)\; \middle|\; \neg\mc{B}_2^{k,K}(\rho)}\label{eq:smsl4}
\end{align}
where the first inequality follows from \expref{Claim}{clm:B2kell} and we ensure the final inequality by choosing $m$ so that $(epK/k)^k \leq (epL/k)^k \leq \sqrt{p}$. This can be done by choosing $m = O(\log(1/p)\cdot (1/p)^{1/2k})$ since we have
\[
\left(\frac{epL}{k}\right)^k \leq (3pL)^k = \left(\frac{300p\log^2(1/\varepsilon)}{\varepsilon^2}\right)^k \leq \left(\frac{300p\log^2(1/p)}{m^2 p}\right)^k\leq \left(\frac{300\log^2(1/p)}{m^2}\right)^k\leq \sqrt{p}
\]
for $m$ as chosen above. Note that $m = (1/p)^{o(1)}$ since we have $k = \omega(1)$ by assumption. 

The event $\neg\mc{B}_2^{k,K}(\rho)$ only depends on the choice of the sub-restriction $\rho|_{[K]}$ and we can condition on $\rho|_{[K]}$ so that this event occurs. {}From now on, the random choice will be a restriction $\rho'\sim \mc{R}_p^{n-K}$ on the remaining variables. 

Let $S$ denote the set of variables in $I$ among the first $K$ variables. We show that with high probability over the choice of $\rho'$, we obtain a $t$-imbalanced threshold gate for each setting of the variables in $S$.  

There are $2^k$ settings to the variables in $S$. For each such setting, the restricted linear function is now $\varepsilon$-regular by the definition of the $\varepsilon$-critical index. Hence we can appeal to the regular case (\ie, inequality (\ref{eq:smsl1})) and use a union bound over the settings to the variables in $S$ to conclude that 
$$\prob{\rho'}{\mc{B}^t(\rho)\; \middle|\; \neg\mc{B}^{k,K}_2(\rho)}\leq 2^k p^{1/2 - o(1)}\leq p^{1/2-o(1)}$$ where we have used our assumption that $k = o(\log(1/p))$ for the final inequality. Along with (\ref{eq:smsl4}), this implies the lemma in the case that $K\leq L$.

In the case that $K > L$, we use a slightly different strategy.
\begin{align}
\prob{\rho}{\tilde{\mc{B}}^t(\rho)} &\leq \prob{\rho}{\mc{B}_2^{k,L}(\rho)} + \prob{\rho}{\tilde{\mc{B}}^t(\rho)\; \middle|\; \neg\mc{B}_2^{k,L}(\rho)}\notag\\
&\leq (epL/k)^k + \prob{\rho}{\tilde{\mc{B}}^t(\rho)\; \middle|\; \neg\mc{B}_2^{k,L}(\rho)} \leq \sqrt{p} + \prob{\rho}{\tilde{\mc{B}}^t(\rho)\; \middle|\; \neg\mc{B}_2^{k,L}(\rho)}.\label{eq:smsl5}
\end{align}

As we did above, we condition on a fixed $I$ so that $\neg\mc{B}_2^{k,L}(\rho)$ occurs (\ie, at most $k$ among the first $L$ variables belong to $I$). Let $S$ denote the set of variables in $I$ among the first $L$ variables. Note that $|S|\leq k$.

Now, note that if $\tilde{\mc{B}}^t(\rho)$ does occur, then on a uniformly random setting to the variables in $S$, the probability that the restricted threshold gate is not $t$-imbalanced is at least $1/2^k$.

On the other hand, we analyze what happens when we set all variables not in $I$ and \emph{also the variables in $S$} uniformly at random. Note that \expref{Proposition}{prop:DGJSV} is now applicable with $L' = L$ and $r = 10\log(1/\varepsilon)$ since we are setting all of the first $L$ variables (and also others). Applying \expref{Proposition}{prop:DGJSV}, the probability that the restricted threshold gate is not $t$-imbalanced is at most $\varepsilon^{10} < p$. Along with the previous paragraph, this implies that $$\prob{}{\tilde{\mc{B}}^t(\rho)\; \middle|\; \neg\mc{B}_2^{k,L}(\rho)}\leq p\cdot 2^k \leq \sqrt{p}$$ since $k = o(\lg(1/p))$.

Putting this together with (\ref{eq:smsl5}), we have the claimed upper bound on $\tilde{\mc{B}}^t(\rho)$ in the case that $K > L$.
\end{proof}

\subsection{Superlinear gate lower bounds and superquadratic wire lower bounds for threshold circuits}

We now prove \expref{Theorem}{thm:KW}. The proof closely follows~\cite{KW} except that we use the structural lemma proved above in lieu of the random restriction lemma from~\cite{KW}.

The following proposition (a weaker version of which is implicit in~\cite{KW}) will help us prove lower bounds on circuits with an unweighted majority gate on top. 

\begin{proposition}
\label{prop:unweighMaj}
Let $h$ be a Boolean function computed by an unweighted majority gate of fan-in $s\geq 1$. Let $f,g_1,\ldots,g_s$ be Boolean functions of $n$ variables such that $\max_{i\in [s]}\Corr(f,g_i)\leq ({1}/{16s})$ and $\Corr(f,1) \leq ({1}/{16s})$ where $1$ denotes the constant function that takes value $1$ at all inputs. Then, $\Corr(f,h(g_1,\ldots,g_s)) \leq 1-({1}/{8s})$.
\end{proposition}

\begin{proof}
  We can write $h(x_1,\ldots,x_s) = \sgn(\sum_{i\leq s}w_i x_i -\theta)$ where $|w_i| = 1$ for each $i$. We can also assume that $|\theta|\leq s+1$ since otherwise we may replace $\theta$ by $\theta' := \sgn(\theta)\cdot (s+1)$ without changing the function. Similarly, by changing $\theta$ if necessary, we can assume that $\theta$ is an integer multiple of ${1}/{2}$. Let
  \[
    L'(y) := \sum_{i\leq s}w_i g_i(y) -\theta
  \]
  and $h'(y) = h(g_1(y),\ldots,g_s(y))$.

{}From the fact that $\Corr(f,g_i)\leq {1}/{16s}$ for all $i\in [s]$ and $\Corr(f,1) \leq {1}/{16s}$, we obtain
\begin{equation}
\label{eq:unwMaj}
|\ip{f}{L'}|\leq \sum_{i}|w_i|\cdot |\ip{f}{g_i}| + |\theta|\cdot |\ip{f}{1}| \leq s\cdot \frac{1}{16s} + \frac{(s+1)}{16s}< \frac{1}{4}.
\end{equation}

Assume for the sake of contradiction that $\Corr(f,h') > 1-({1}/{8s})$. Then, we have
\[
  \max\left\{\prob{y\in \{-1,1\}^n}{f(y) = h'(y)},\prob{y\in \{-1,1\}^n}{f(y) \neq h'(y)}\right\}\geq 1-({1}/{16s}).
\]
We assume that $\prob{y\in \{-1,1\}^n}{f(y) = h'(y)}\geq 1-({1}/{16s})$ and obtain a contradiction to (\ref{eq:unwMaj}) (the other case is similar).

It is easy to check using our assumptions on the $w_i$ and $\theta$ that whenever $f(y) = h'(y)$, then $f(y)L'(y) \geq {1}/{2}$. Also, for any $y$ such that $f(y)\neq h'(y)$, we have $f(y)L'(y)\geq -2s-1$. Since
\[
  \prob{y}{f(y) = h'(y)}\geq 1-({1}/{16s}),
\]
we get
\begin{align*}
|\ip{f}{L'}|\geq \ip{f}{L'} &\geq \prob{y}{f(y) = L'(y)}\cdot\frac{1}{2} - \prob{y}{f(y) \neq L'(y)}\cdot(2s+1)\\
&\geq \left(1-\frac{1}{16s}\right)\cdot \frac{1}{2} - \frac{(2s+1)}{16s} > \frac{1}{4}.
\end{align*}
This contradicts (\ref{eq:unwMaj}).
\end{proof}

Recall the statement of \expref{Theorem}{thm:KW}.  

\begin{reptheorem}{thm:KW}[Restated]
For infinitely many $n\in\mathbb{N}$, there is an explicit (polynomial-time computable) function  defined on $2n$ variables that does not have $\MAJ\circ\THR\circ\THR$ circuits with fewer than $n^{1.5-o(1)}$ gates or $n^{2.5-o(1)}$ wires.
\end{reptheorem}

The explicit function we will use is the same as that defined by Kane and Williams~\cite[Section 6]{KW}, which is  also a generalization of the Andreev function. The Generalized Andreev function from \expref{Section}{sec:andreev} can also be suitably modified (by changing the parameters) to yield the same lower bounds.

To prove the above theorem, we need some definitions and some lemmas which are closely related to  statements proved in~\cite{KW}. 

For positive integers $m,k,\ell$, let $\THR\circ\THR[m, k,\ell]$ denote the class of $\THR\circ \THR$ circuits on $m$ variables with at most $\ell$ gates on the bottom layer that have fan-in greater than $k$. Similarly, let $\MAJ\circ \THR\circ\THR[m,k,\ell,s]$ denote the class of $\THR\circ \THR$ circuits on $m$ variables with at most $\ell$ gates on the bottom layer that have fan-in greater than $k$ and at most $s$ gates on the second layer.

\begin{lemma}
\label{lem:KW-count}
Let $m,k,\ell$ be arbitrary integers.  The number of distinct Boolean functions in $\THR\circ\THR[m, k,\ell]$ is bounded by $2^{r}$, where $r = O(2^{O(k^2)}\cdot m^{k+1} + m^2 \ell)$.
\end{lemma}

\begin{proof}
We show that each circuit from the class $\THR\circ\THR[m, k,\ell]$ can be described using $r$ bits, where $r$ is as above. This will prove the claim. 

To see this, note that by \expref{Corollary}{cor:numthr} the number of distinct threshold functions of $m$ bits is at most $2^{O(m^2)}$ and hence, all threshold gates of fan-in greater than $k$ can be described using at most $O(m^2\ell)$ bits. The gates of fan-in at most $k$ can compute at most
\[
  \binom{m}{k}\cdot 2^{O(k^2)}\leq m^k 2^{O(k^2)}
\]
distinct functions; we assume that all these functions appear in the
bottom layer.

The threshold gate on top thus computes a threshold function of at most $\ell + m^k 2^{O(k^2)}$ threshold functions at the bottom layer. By \expref{Theorem}{thm:RSO94}, each such function can be described using
\[
  O\left(\bigl(\ell + m^k 2^{O(k^2)}\bigr)\cdot m\right)\leq O\left(\ell m + 2^{O(k^2)}\cdot m^{k+1}\right)
\]
many bits.

Thus, the total description length is $O(\ell m^2 + 2^{O(k^2)}\cdot m^{k+1})$ as required.
\end{proof}

We need the following lemma, implicit in~\cite{KW}.

\begin{lemma}[Implicit in Corollary 6.2 from~\cite{KW}]
\label{lem:KW-F}
Fix any parameters $M,n\in\mathbb{N}$ such that $M$ is a power of $2$ and $M \leq 2^{n/2}$. There is an explicit polynomial-time computable function $F:\{-1,1\}^{\log M} \times \{-1,1\}^n\rightarrow\{-1,1\}$ such that given any class of Boolean functions $\mc{F}$ on $\log M$ Boolean variables with $|\mc{F}|\leq 2^{o(n)}$, there is an $x\in \{-1,1\}^n$ such that for any $f\in \mc{F}$, we have $\Corr(f,F(\cdot,x))\leq O(\sqrt{n})/M^{1/4}$.
\end{lemma}

Using the function $F$ from the above lemma, we define our hard function, as in~\cite{KW}, in the following way. Set $M = n^{16}$ in the above lemma. In what follows, we assume that $n$ is a power of $2$ and that $\log M\ |\ n$. The function $B_{n}:\{-1,1\}^{n}\times \{-1,1\}^n\rightarrow\{-1,1\}$ is defined to be $B_n(x,y) = F(z,x)$ where $z\in \{-1,1\}^{\log M}$ is obtained by partitioning the $y$ variables into $\log M$ blocks of size $n/\log M$ and taking the parity of the bits in the $i$th block to obtain $z_i$ ($i\in [\log M]$).

Lemmas~\ref{lem:KW-count} and~\ref{lem:KW-F} imply a strong correlation bound for $B_n$ against certain kinds of circuits. We will say that a restriction on the $y$ variables of $B_n$ is \emph{live} if it leaves at least one variable unset in each of the $\log M$ blocks defined above. 

\begin{lemma}
\label{lem:Bn}
Let $n\in\mathbb{N}$ and $M=n^{16}$ be as above. Fix any $m,k,\ell\in\mathbb{N}$ such that the number of functions in $\THR\circ \THR[\log M,k,\ell]$ is $2^{o(n)}$. Then there is an $x\in \{-1,1\}^n$ such that for any live restriction $\rho$ on the $y$ variables that leaves $m$ variables unset, the restricted function $B_{n,\rho}:= B(x,\cdot)|_\rho$ satisfies $\Corr(C,B_{n,\rho})\leq 1-\Omega\left({1}/{n^3}\right)$ for each circuit $C$ from the class $\MAJ\circ\THR\circ\THR[m,k,\ell,n^3]$.
\end{lemma}

\begin{proof}
By \expref{Lemma}{lem:KW-F}, we can fix an $x\in \{-1,1\}^n$ such that the function $F(\cdot,x)$ has correlation at most $O(1/n^{3.5})$ with any function from the class $\THR\circ \THR[\log M,k,\ell]$. \expref{Proposition}{prop:unweighMaj} then implies that any circuit from the class $\MAJ\circ\THR\circ\THR[m,k,\ell,n^3]$ has correlation at most $1-\Omega\left({1}/{n^3}\right)$ with $F$.

For each $i\in [\log M]$, fix a variable $y^{(i)}$ in the $i$th block of $y$ variables that is left unset by $\rho$. Let $Y$ denote the set of $\log M$ variables chosen in this way. Construct a restriction tree $T$ that sets all the variables not in $Y$ that are not set by $\rho$. For each leaf $\lambda$ of $T$, we let $C_\lambda$ denote the circuit $C|_{\rho_\lambda}$ and similarly $B_{n,\lambda}$ for $B_{n,\rho}|_{\rho_\lambda}$. By \expref{Fact}{fac:corr}, we have
\[
  \Corr(C,B_{n,\rho}) \leq \avg{\lambda\sim T}{\Corr(C_\lambda,B_{n,\lambda}}
\]
and in particular there is a $\lambda$ such that
$\Corr(C,B_{n,\rho})\leq \Corr(C_\lambda,B_{n,\lambda})$. Fix any such
$\lambda$ for the rest of the proof.

Note that $C_\lambda$ is a circuit from the class $\MAJ\circ\THR\circ\THR[m,k,\ell,n^3]$. Also note that $B_{n,\lambda}$ is exactly a copy of the function $F$, with some subset of the variables of $F$ being replaced by their negations. By negating some of the variables in $C_\lambda$ if necessary, we can obtain a circuit $C'_\lambda$ from the class $\MAJ\circ\THR\circ\THR[m,k,\ell,n^3]$ such that
\[
\Corr(C_\lambda,B_{n,\lambda}) = \Corr(C'_\lambda, F) \leq 1-\Omega\left(\frac{1}{n^3}\right)
\]
where the latter inequality follows from our reasoning above. By our choice of $\lambda$, we obtain
\[
  \Corr(C,B_{n,\rho})\leq\Corr(C_\lambda,B_{n,\lambda})\leq 1-\Omega({1}/{n^3}),
\]
which concludes the proof of the lemma.
\end{proof}

\begin{proof}[Proof of \expref{Theorem}{thm:KW}]
The proof follows closely the ideas of~\cite[Theorem 1.3]{KW}. The main new ingredient is the use of \expref{Lemma}{lem:thresh-gate-2}, which replaces the use of the random restriction lemmas from~\cite{KW}.

Fix $m,k\in \mathbb{N}$ such that $\omega((\log n)^2) \leq m\leq  n^{o(1)}$, $2^{k^2} \leq n^{o(1)}$ and $m^k = n^{\omega(1)}$. (For example, we could set $m = 2^{\lceil (\log n)^{5/6} \rceil}$ and $k = \lceil (\log n)^{1/4}\rceil$.) Let $\ell = n/m^3$.

By our choice of parameters and \expref{Lemma}{lem:KW-count}, the number of functions in $\THR\circ\THR[\log M, k, \ell]$ is $2^{o(n)}$. By \expref{Lemma}{lem:Bn}, we can find an $x$ so that the restricted function $B_{n}(x,\cdot)$ is as guaranteed by \expref{Lemma}{lem:Bn}. We fix any such $x$ for the rest of the proof and consider the restricted function $B_n(x,\cdot)$.

We begin with the gates case. Let $C$ be a $\MAJ\circ \THR\circ \THR$ circuit computing $B_{n}(x,\cdot)$ with at most $s = n^{1.5}/m^{10}r$ gates where $r \leq n^{o(1)}$ is a parameter that will be chosen below.

We will apply a random restriction $\rho = (I,w)\sim \mc{R}_p^n$ where $p = m/n$. We bound the probability of the following bad events:

\begin{enumerate}
\item Event $\mc{E}_1(\rho)$: This is the event that the restriction $\rho$ is not live. For this event to occur, there is some $i\in [\log M]$ such that each of the bits of $y$ that $z_i$ depends on must be fixed to a constant by $\rho$. For each $i$, the probability that this happens is at most $(1-m/n)^{n/\log M}\leq 1/2^{\Omega(m/\log M)}\leq 1/n^{\omega(1)}$. By a union bound over $i\in [\log M]$, the probability of $\mc{E}_1$ is at most $1/n^{\omega(1)}$.

\item Event $\mc{E}_2(\rho)$: By \expref{Lemma}{lem:thresh-gate-2}, the probability that any threshold gate is not $(t,k)$-imbalanced for $t = \log n$ is at most $p^{1/2-\delta}$ where $\delta\rightarrow 0$ as $n\rightarrow\infty$. As $p\geq 1/n$, we may bound this probability by $\sqrt{p}\cdot n^{\delta}$. Set $r = n^{\delta}$. 

The expected number of gates that are not $(t,k)$-imbalanced is at most $s\cdot\sqrt{p}\cdot r\leq n/m^9$. The event $\mc{E}_2$ is that the number of gates that are not $(t,k)$-imbalanced is at least $2n/m^9$. By Markov's inequality, the probability of this event is at most $1/2$.
\end{enumerate}

By a union bound, there exists a restriction $\rho$ such that neither of the events $\mc{E}_1$ nor $\mc{E}_2$ occurs. Fix such a restriction $\rho$ and consider the circuit $C|_\rho$. The number of gates that are not $(t,k)$-imbalanced in $C|_\rho$ is bounded by $2n/m^9\leq \ell$. 

\begin{sloppypar}
  By \expref{Fact}{fact:tkimbal}, any $(t,k)$-imbalanced threshold gate $\phi$ in $C|_\rho$ can be approximated to error $\exp(-\Omega(t^2)) = 1/n^{\omega(1)}$ by a threshold gate $\phi'$ of fan-in at most $k$. By a union bound, the circuit $C'_\rho$ obtained by replacing each $(t,k)$-imbalanced $\phi$ by a $\phi'$ as defined above satisfies
\[
\prob{y'\in \{-1,1\}^{|I|}}{C|_\rho(y')\neq C'_\rho(y')} \leq s\cdot \frac{1}{n^{\omega(1)}}\leq \frac{1}{n^{\omega(1)}}.
\]

In particular, since $C|_\rho$ computes the function $B':= B_{n}(x,\cdot)|_\rho$ we have
\[
\prob{y'\in \{-1,1\}^{|I|}}{B'(y')\neq C'_\rho(y')} \leq \frac{1}{n^{\omega(1)}}.
\]
and hence $\Corr(B',C'_\rho) \geq 1-({1}/{n^{\omega(1)}})$ which contradicts \expref{Lemma}{lem:Bn}. This proves the theorem for the case of circuits with at most $n^{1.5-o(1)}$ gates.
\end{sloppypar}

We now consider the wires case, which is very similar. Let $C$ now be a $\MAJ\circ \THR\circ \THR$ circuit computing $B_{n}(x,\cdot)$ with at most $s = n^{2.5}/m^{10}r$ \emph{wires} where $r$ is as defined above. We will say that a gate on the bottom layer is \emph{large} if its fan-in is at least $n/m^2$ and \emph{small} otherwise. We use $L$ to denote the set of large gates and $S$ to denote the set of small gates. Let $s_L := |L|$. Note that $s_L\leq s/(n/m^2) = n^{1.5}/m^8 r$. 

As in the gates case, we apply a random restriction $\rho\sim \mc{R}_p^n$ to the circuit $C$. We now consider the following bad events.

\begin{enumerate}
\item Event $\mc{E}_1(\rho)$: This event is as defined above. As noted above, the probability of this event is bounded by $1/n^{\omega(1)}$.
\item Event $\mc{E}_2'(\rho)$: This is the event that some gate in $S$ has fan-in at least $k$ \emph{after} the restriction. Fix any threshold gate $\phi\in S$. Since $\phi$ has fan-in at most $n/m^2$, the probability that it has fan-in at least $k$ after the restriction is bounded by
\[
\binom{n/m^2}{k} \cdot \left(\frac{m}{n}\right)^k \leq \left(\frac{enm}{knm^2}\right)^k \leq \frac{1}{m^{k}} = \frac{1}{n^{\omega(1)}}
\]
where the last equality follows from our choice of parameters. Since the circuit $C$ has at most $s$ wires, it also has at most $s$ gates and hence by a union bound over the gates in $S$, the probability that $\mc{E}_1'(\rho)$ occurs is at most $n^{2.5}/n^{\omega(1)}\leq 1/n^{\omega(1)}$.

\item Event $\mc{E}_3'(\rho)$: This is the event that the number of gates that are not $(t,k)$-imbalanced is at most $2n/m^{7}$. Similar to the gates case, we can argue that the expected number of gates in $L$ that are not $(t,k)$-imbalanced is at most $s_L\sqrt{p}\cdot r\leq n/m^{7}$. Thus, by Markov's inequality, we see that the probability that $\mc{E}'_3(\rho)$ occurs is at most $1/2$.
\end{enumerate}

By a union bound, we see that there is a $\rho$ such that none of $\mc{E}_1(\rho),\mc{E}_2'(\rho),$ or $\mc{E}_3'(\rho)$ occur. {}From here on, we can follow the proof of the gates case verbatim to derive a contradiction to \expref{Lemma}{lem:Bn}. This proves the theorem in the wires case.

\end{proof}

\paragraph*{Acknowledgements.} A major portion of this work was carried out while visiting the Simons Institute for Theory of Computing at Berkeley. We thank the Institute for their support. 
The first and second authors conducted the work when the authors were at University of Edinburgh, and were supported by the European Research Council under the European Union's
Seventh Framework Programme (FP7/2007-2013)/ ERC Grant Agreement no. 615075.
The third author would like to thank the School of Informatics at the University of Edinburgh for hosting him in March-April 2015, when this research project was initiated. 

We would also like to thank the anonymous CCC 2016 reviewers for their feedback and suggestions and also for pointing out to us the
article by  %
Janson~\cite{Janson}. We also thank anonymous reviewers for
\emph{Theory of Computing} and Prahladh Harsha for their
detailed feedback on the write-up and corrections.

\bibliographystyle{tocplain}   %
\bibliography{v014a009}

\begin{tocauthors}
\begin{tocinfo}[chen]
Ruiwen Chen\\
rwchenmail\tocat{}gmail\tocdot{}com \\
\end{tocinfo}
\begin{tocinfo}[santhanam]
Rahul Santhanam\\
Department of Computer Science\\
University of Oxford\\
Oxford, U.K.\\   %
rahul.santhanam\tocat{}cs\tocdot{}ox\tocdot{}ac\tocdot{}uk \\
\end{tocinfo}
\begin{tocinfo}[srinivasan]
  Srikanth Srinivasan\\
  Department of Mathematics\\
  IIT Bombay, Mumbai, India\\
  srikanth\tocat{}math\tocdot{}iitb\tocdot{}ac\tocdot{}in \\
\end{tocinfo}
\end{tocauthors}

\begin{tocaboutauthors}

\begin{tocabout}[chen]
\textsc{Ruiwen Chen} graduated with 
a \phd\  %
in Computer Science %
from Simon Fraser University in 2014,
advised by Valentine Kabanets, and then worked as a
postdoc   %
with Rahul Santhanam until 2017.
\end{tocabout}

\begin{tocabout}[santhanam]
\textsc{Rahul Santhanam}
got his \phd\
in Computer Science %
from the University of Chicago in 2005,
advised by Lance Fortnow and Janos Simon. After stints at Simon Fraser
University, the University of Toronto and the University of Edinburgh,
he is now a Professor of Computer Science at the University of Oxford. His
main research interest is in complexity lower bounds.
\end{tocabout}

\begin{tocabout}[srinivasan]
\textsc{Srikanth Srinivasan} 
got his undergraduate degree from the
\href{http://www.iitm.ac.in}{Indian Institute of Technology Madras}.
He received  %
his \phd\ from
\href{http://www.imsc.res.in}{The Institute of Mathematical Sciences}
in Chennai, India %
in 2011, where his advisor was
\href{http://www.imsc.res.in/~arvind}{V. Arvind}. 
His research interests include circuit complexity, derandomization, and
related areas of mathematics. 
\end{tocabout}
\end{tocaboutauthors}

\end{document}